\documentclass[sigplan,screen,nonacm]{acmart}

\usepackage[T1]{fontenc}
\usepackage[english]{babel}

\usepackage[utf8]{inputenc}
\usepackage{amsfonts,amsmath,amsthm}
\usepackage[all,arc]{xy}
\usepackage{enumerate}
\usepackage{mathrsfs}
\usepackage{comment}
\usepackage{graphicx} 
\usepackage{commath}
\usepackage{enumerate}
\usepackage{afterpage}
\usepackage{xcolor}
\usepackage{listings}
\usepackage{paralist}
\usepackage{caption}
\usepackage{subcaption}
\usepackage{natbib}
\usepackage{microtype}
\usepackage[ruled, linesnumbered,vlined]{algorithm2e}

\theoremstyle{remark}
\newtheorem{rmk}{Remark}

\newenvironment{ex}{\textbf{Example}}{}

\lstdefinelanguage{affprob}
{
	morekeywords={angel,demon, choice, prob(0.6), prob(0.5), if, then, else, fi,
		while, do, od,
		true, false, and, or, skip, sample},
	sensitive = false
}

\usepackage{tikz}
\usepackage{tikz,pgffor}
\usepackage{changepage}
\usetikzlibrary{arrows}
\usetikzlibrary{shapes}
\usetikzlibrary{calc}
\usetikzlibrary{automata}
\usetikzlibrary{positioning}
\usetikzlibrary{angles}

\tikzstyle{ang}=[regular polygon, regular polygon sides = 3,draw,inner sep=0pt,minimum size=6mm, yshift = -0.75 mm]
\tikzstyle{dem}=[shape=diamond,draw,inner sep=0pt,minimum size=6mm]
\tikzstyle{ran}=[shape=circle,draw,inner sep=0pt,minimum size=6mm]
\tikzstyle{det}=[shape=rectangle,draw,inner sep=0pt,minimum size=5mm]
\tikzstyle{tran}=[draw,->,>=stealth, rounded corners]

\usepackage{todonotes}

\newcommand{\vars}{\mathcal{V}}

\newcommand{\CFG}{\mathcal{C}}
\newcommand{\trsys}{\mathcal{T}}

\newcommand{\locinit}{\loc_{\mathit{init}}}

\newcommand{\updates}{\mathit{Up}}

\newcommand{\locterm}{\loc_\lout}

\newcommand{\locs}{\mathit{L}}

\newcommand{\loc}{\ell}

\newcommand{\lout}{\mathit{out}}

\newcommand{\transitions}{\mapsto}

\newcommand{\guards}{G}

\newcommand{\conf}{c}

\newcommand{\BI}{\mathit{BI}}
\newcommand{\NA}{\mathit{NA}}
\newcommand{\Rev}{\text{Rev}}

\newcommand{\TermComp}{\mathsf{TermComp`19}}
\newcommand{\LoAT}{\mathsf{LoAT}}
\newcommand{\AProVE}{\mathsf{AProVE}}
\newcommand{\Ultimate}{\mathsf{Ultimate}}
\newcommand{\VeryMax}{\mathsf{VeryMax}}
\newcommand{\StarExec}{\mathsf{StarExec}}
\newcommand{\RevTerm}{\mathsf{RevTerm}}



\begin{document}

\title{Proving Non-termination by Program Reversal}

\author{Krishnendu Chatterjee}
\affiliation{IST Austria, Klosterneuburg, Austria}
\email{krishnendu.chatterjee@ist.ac.at}

\author{Ehsan Kafshdar Goharshady}
\affiliation{Ferdowsi University of Mashhad, Mashhad, Iran}
\email{e.goharshady1@gmail.com}

\author{Petr Novotn\'y}
\affiliation{Masaryk University, Brno, Czech Republic}
\email{petr.novotny@fi.muni.cz}

\author{\DJ or\dj e \v{Z}ikeli\'c}
\affiliation{IST Austria, Klosterneuburg, Austria}
\email{dzikelic@ist.ac.at}

\begin{abstract}

We present a new approach to proving non-termination of non-deterministic integer programs. Our technique is rather simple but efficient. It relies on a purely syntactic reversal of the program's transition system followed by a constraint-based invariant synthesis with constraints coming from both the original and the reversed transition system. The latter task is performed by a simple call to an off-the-shelf SMT-solver, which allows us to leverage the latest advances in SMT-solving. Moreover, our method offers a combination of features not present (as a whole) in previous approaches: it handles programs with non-determinism, provides relative completeness guarantees and supports programs with polynomial arithmetic. The experiments performed with our prototype tool $\RevTerm$ show that our approach, despite its simplicity and stronger theoretical guarantees, is at least on par with the state-of-the-art tools, often achieving a non-trivial improvement under a proper configuration of its parameters.
\end{abstract}

\maketitle

\section{Introduction}

{\em Program analysis.} 
There are two relevant directions in program analysis: to prove 
program correctness and to find bugs.
While a correctness proof is obtained once, the procedure of bug finding 
is more relevant during software development and is repeatedly
applied, even for incomplete or partial programs.
In terms of specifications, the most basic properties in program analysis 
are safety and liveness.

\smallskip\noindent{\em Program analysis for safety and termination.}
The analysis of programs with respect to safety properties has received
a lot of attention~\citep{BallR02,HenzingerJMS02,GodefroidKS05,GulavaniHKNR06}, and for safety properties to report errors
the witnesses are finite traces violating the safety property.
The most basic liveness property is termination.
There is a huge body of work for proving correctness with respect to the termination property~\cite{FrancezGKP85,ColonS02,BradleyMS05,Cousot05},
e.g.~sound and complete methods based on ranking functions have been developed~\cite{ColonS01,PodelskiR04,PodelskiR04trin},
and efficient computational approaches based on lexicographic ranking functions have also
been considered~\cite{BradleyMS05,CookSZ13,BrockschmidtCF13}.

\smallskip\noindent{\em Proving non-termination.}
The bug finding problem for the termination property, or proving non-termination, is a 
challenging problem.
Conceptually, while for a safety property the violating witness is a finite trace,
for a termination property the violating witnesses are infinite traces.
There are several approaches for proving non-termination; here we discuss some key ones, which are most related in spirit to our new method (for a detailed discussion of  related work, see Section~\ref{sec:relatedwork}). For the purpose of this overview, we (rather broadly and with a certain grain of salt) classify the approa\-ches into two categories: \emph{trace-based} approaches, which look for a non-terminating trace (e.g.~\cite{GuptaHMRX08,LeikeH18,FrohnG19}), and \emph{set-based app\-roa\-ches,} which look for a set of non-terminal program configurations (states) in which the program can stay indefinitely (e.g.~\cite{ChenCFNO14,LarrazNORR14,GieslABEFFHOPSS17}). 
For instance, the work of~\cite{GuptaHMRX08} considers computing "lassos" (where a lasso is a finite prefix
followed by a finite cycle infinitely repeated) as counter-examples for 
termination and presents a trace-based app\-roach based on lassos to prove non-termination 
of deterministic programs.
In general, finite lassos are not sufficient to witness non-termination.
While lassos are periodic, proving non-termination for programs with aperiodic infinite traces 
via set-based methods has
 been considered in~\cite{ChenCFNO14,LarrazNORR14} for programs with non-determinism.
In~\cite{ChenCFNO14}, a method is proposed where "closed recurrence sets" of configurations are used to 
prove non-termination. Intuitively, a closed recurrence set must contain some initial configuration, must contain no terminal configurations, and cannot be escaped once entered. In \cite{ChenCFNO14}, closed recurrence sets are defined with respect to under-approximations of the transition 
relation, and an under-approximation search guided by several calls to a safe\-ty prover is used to compute a closed recurrence set.
In~\cite{LarrazNORR14}, a constraint solving-based method is proposed to search for ''quasi-invariants'' (sets of configurations which cannot be left once entered)
exhaustively in all strongly-connected subgraphs. A safety prover is used to check 
reachability for every obtained quasi-invariant.
For constraint solving, Max-SMT is used in~\cite{LarrazNORR14}.

\smallskip\noindent{\em Limitations of previous approaches.}
While the previous works represent significant advancement for proving 
non-termina\-tion, each of them has, to our best knowledge, at least one of the following limitations:
\begin{compactenum}[a)]
\item They do not support non-determinism, e.g.~\cite{VelroyenR08,GulwaniSV08}.
\item They only work for lassos (i.e. periodic non-terminating traces), e.g.~\cite{GuptaHMRX08}.
\item \emph{Theoretical limitation} of not providing any \emph{(relative) completeness guarantees.} Clearly, a non-termination proving algorithm cannot be both sound and complete, since non-termination is well-known to be undecidable.
However, as in the case of termination proving, it can be beneficial to provide relative completeness guarantees, i.e. conditions on the input program under which the algorithm is guaranteed to prove non-termination. To our best knowledge, the only approaches with such guarantees are \cite{LeikeH18,GulwaniSV08}; however, both of them only provide guarantees for a certain class of \emph{deterministic} programs.
\item Most of the previous approaches do not support programs with polynomial arithmetic (with an exception of \cite{CookFNO14,FrohnG19}).
\end{compactenum}

\smallskip\noindent{\em Our contributions.}
In this work we propose a new set-based approach to non-termination proving in integer programs. Intuitively, it searches for a diverging program configuration, i.e.~a configuration that is reachable but from which no program run is terminating (after resolving non-determinism using symbolic polynomial assignments). Our approach is based on a simple technique of {\em program reversal}, which reverses each transition in the program's transition system to produce the reversed transition system. The key property of this construction is that, given a program configuration, there is a terminating run starting in it if and only if it is reachable from the terminal location in the reversed transition system. This allows over-approximating the set of all program configurations from which termination can be reached by computing an invariant in the program's reversed transition system. We refer to the invariants in reversed transition systems as {\em backward invariants}. To generate the backward invariant, we may employ state-of-the-art polynomial invariant generation techniques to the reversed transition system as a single-shot procedure which is the main practical benefit of the program reversal. Our method proves non-termination by generating a backward invariant whose complement is reachable.
Hence, our new method adapts the classical and well-studied techniques for inductive invariant generation in order to find non-termination proofs by combining forward and backward analysis of a program. While such a combined analysis is common in safety analysis where the goal is to show that no program run reaches some annotated set of configurations~\cite{Bourdoncle93}, to our best knowledge it has never been considered for proving non-termination in programs with non-determinism, where we need to find a single program run that does not terminate. The key features of our method are as follows:

\begin{compactenum}[a)]
\item Our approach supports programs with non-determinism.
\item Our approach is also applicable to programs where all non-terminating traces are aperiodic.
\item {\em Relative completeness guarantee:} The work of~\cite{ChenCFNO14} establishes that 
closed recurrence sets and under-approxima\-tions are a sound and complete \emph{certificate} of non-termina\-tion, yet the algorithm based on these certificates does not in itself
provide any relative completeness guarantee (in the above sense). 
For our approach we show the following:
If there is an under-approximation of the transition relation where non-determinism can be resolved by polynomial assignments such that the resolved program contains a closed recurrence set representable as a propositional predicate map, 
then our approach is guaranteed to prove non-termination.
We obtain such guarantee by employing relatively complete methods for inductive invariant synthesis, which is another key advantage of adapting invariant generation techniques to non-termination proving.
Moreover, we provide even stronger relative completeness guarantees for programs in which non-determi\-nism appears only in branching (but not in variable assignments).
\item Our approach supports programs with polynomial arithmetic.
\end{compactenum}
We developed a prototype tool $\RevTerm$ which implements our approach. We experimentally compared our tool with state-of-the-art non-termination provers on standard benchmarks from the Termination and Complexity Competition ($\TermComp$~\citep{GieslRSWY19}).
Our tool demonstrates performance on par with the most efficient of the competing provers, while providing additional guarantees. 
In particular, with a proper configuration, our tool achieved the largest number of benchmarks proved non-terminating.


\smallskip\noindent{\em Outline.}
After presenting the necessary definitions (Section~\ref{sec:prelims}), we present our approach and its novel aspects in the following order: first we introduce the technique of program reversal (Section~\ref{sec:reverse}); then we present a new certificate for non-termination (Section~\ref{sec:certificate}) based on so-called \emph{backward invariants,} 
as well as an invariant generation-based automated approach for this certificate 
(Section~\ref{sec:nonterm}); finally
we prove relative completeness guarantees (Section~\ref{sec:completeness}). We conclude with the presentation of our experiments and discussion of related work.

 

\section{Preliminaries}\label{sec:prelims}

\noindent{\em Syntax of programs.} In this work we consider simple imperative arithmetic programs with polynomial integer arithmetic and with non-determinism. They consist of standard programming constructs such as conditional branching, while-loops and (deterministic) variable assignments. In addition, we allow constructs for non-deterministic assignments of the form $x:=\textbf{ndet}()$, which assign any integral value to $x$. The adjective \textit{polynomial} refers to the fact that all arithmetic expressions are polynomials in program variables.

\begin{example}[Running example]\label{ex:running}
	Fig.~\ref{fig:running} left shows a program which will serve as our running example. The second line contains a non-deterministic assignment, in which any integral value can be assigned to the variable $x$.
\end{example}

\noindent{\em Removing non-deterministic branching.}
We may \textit{without loss of generality} assume that non-determinism does not appear in branching: for the purpose of termination analysis, one can replace each non-deterministic branching with a non-deterministic assignment. Indeed, non-deterministic branching in programs is given by a command $\textbf{if } \ast \textbf{ then}$, meaning that the control-flow can follow any of the two subsequent branches. By introducing an auxiliary program variable $x_{\textbf{ndet}}$ and replacing each command $\textbf{if } \ast \textbf{ then}$ with two commands
\begin{equation*}
\begin{split}
&x_{\textbf{ndet}}:=\textbf{ndet}()\\
&\textbf{if } x_{\textbf{ndet}}\geq 0 \textbf{ then}
\end{split}
\end{equation*}
we obtain a program which terminates on every input if and only if the original program does. This removal is done for the sake of easier presentation and neater definition of the \textit{resolution of non-determinism}, see Section~\ref{sec:resnondet}.

\lstset{language=affprob}
\lstset{tabsize=2}
\newsavebox{\exapp}
\begin{lrbox}{\exapp}
	\centering
	\begin{lstlisting}[mathescape]
	$l_0$:	while $x\geq 9$ do
	$l_1$:		$x := \textbf{ndet}()$
	$l_2$:		$y := 10\cdot x$
	$l_3$:		while $x\leq y$ do
	$l_4$:			$x := x+1$
	    od
	   od
	\end{lstlisting}
\end{lrbox}

\smallskip\noindent{\em Predicate, assertion, propositional predicate.} We use the following terminology:
\begin{compactitem}
	\item {\em Predicate}, which is a set of program variable valuations.
	\item {\em Assertion}, which is a finite conjunction of polynomial inequalities over program variables. We need not differentiate between non-strict and strict inequalities since we work over integer arithmetic.
	\item {\em Propositional predicate (PP)}, which is a finite disjunction of assertions.
\end{compactitem}
We write $\mathbf{x}\models \phi$ to denote that the predicate $\phi$ given by a formula over program variables is satisfied by substituting values in $\mathbf{x}$ for corresponding variables in $\phi$. For a predicate $\phi$, we define $\neg \phi=\mathbb{Z}^{|\vars|}\backslash\phi$.

\smallskip\noindent{\em Transition system.} We model programs using transition systems \cite{ColonSS03}.

\begin{definition}[Transition system]\label{def:cfg}
	A \em{transition system} is a tuple $\trsys=(\locs,\vars,\locinit,\Theta_{init},\transitions)$, where
		 $\locs$ is a finite set of {\em locations};
		$\vars$ is a finite set of {\em program variables};
		$\locinit$ is the {\em initial location};
		$\Theta_{init}$ is the set of {\em initial variable valuations}; and
		 $\transitions\, \subseteq \locs\times \locs\times \mathcal{P}(\mathbb{Z}^{|\vars|}\times \mathbb{Z}^{|\vars|})$ is a finite set of {\em transitions}. Each transition is defined as an ordered triple $\tau=(l,l',\rho_{\tau})$, with $l$ its {\em source} and $l'$ the {\em target location}, and the {\em transition relation} $\rho_{\tau}\subseteq \mathbb{Z}^{|\vars|}\times \mathbb{Z}^{|\vars|}$. The transition relation is usually given by an assertion over $\vars$ and $\vars'$, where $\vars$ represents the source-state variables and $\vars'$ the target-state variables.
\end{definition}

\noindent Each program $P$ naturally defines a transition system $\trsys$, with each transition relation given by an assertion over program variables. Its construction is standard and we omit it. The only difference is that here $\locinit$ will correspond to the first non-assignment command in the program code, whereas the sequence of assignments preceding $\locinit$ specifies $\Theta_{init}$ (unspecified variables may take any value). Hence $\Theta_{init}$ will also be an assertion. For transition systems derived from programs, we assume the existence of a special {\em terminal location} $\locterm$, which represents a ''final'' line of the program code. It has a single outgoing transition which is a self-loop with a transition relation $\rho=\{(\mathbf{x},\mathbf{x})\mid \mathbf{x}\in\mathbb{Z}^{|\vars|}\}$.

\noindent A \textit{configuration} (or \textit{state}) of a transition system $\trsys$ is an ordered pair $(l,\mathbf{x})$ where $l$ is a location and $\mathbf{x}$ is a vector of variable valuations. A configuration $(l',\mathbf{x}')$ is a \textit{successor} of a configuration $(l,\mathbf{x})$ if there is a transition $\tau=(l,l',\rho_{\tau})$ with $(\mathbf{x},\mathbf{x}')\in \rho_{\tau}$. The self-loop at $\locterm$ allows us to without loss of generality assume that each configuration has at least one successor in a transition system $\trsys$ derived from a program. Given a configuration $\mathbf{c}$, a \textit{finite path from $\mathbf{c}$} in $\trsys$ is a finite sequence of configurations $\mathbf{c}=(l_0,\mathbf{x}_0),\dots,(l_k,\mathbf{x}_k)$ where for each $0\leq i< k$ we have that $(l_{i+1},\mathbf{x}_{i+1})$ is a successor of $(l_i,\mathbf{x}_i)$. A \textit{run} (or \textit{execution}) \textit{from $\mathbf{c}$} in $\trsys$ is an infinite sequence of configurations whose every finite prefix is a finite path from $\mathbf{c}$. A configuration is said to be {\em initial} if it belongs to the set $\{(\locinit,\mathbf{x})\mid \mathbf{x}\models \Theta_{init} \}$. A configuration $(l,\mathbf{x})$ is \textit{reachable from $\mathbf{c}$} if there is a finite path from $\mathbf{c}$ with the last configuration $(l,\mathbf{x})$. When we omit specifying the configuration $\mathbf{c}$, we refer to a finite path, execution and reachability from some initial configuration. A configuration $(l,\mathbf{x})$ is said to be {\em terminal} if $l=\locterm$.

\begin{example}\label{ex:transitionsystem}
	The transition system for our running example is presented in Fig.~\ref{fig:running} center. It contains 6 locations $\locs=\{l_0,l_1,l_2,l_3,l_4,\locterm\}$ with $\locinit=l_0$, and two program variables $\vars=\{x,y\}$. Since there are no assignments preceding the initial program location, we have $\Theta_{init}=\mathbb{Z}^2$. Locations are depicted by labeled circles, transitions by directed arrows between program locations and their transition relations are given in the associated rectangular boxes.
\end{example}

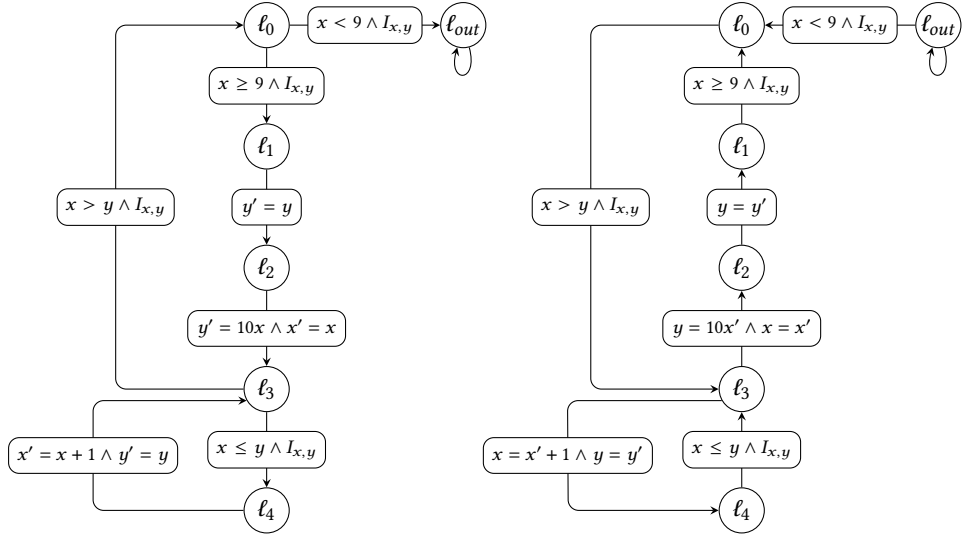
\begin{figure*}[t]
	\centering
	\begin{subfigure}{0.25\textwidth}
		\centering
		\usebox{\exapp}
	\end{subfigure}
	\begin{subfigure}{0.35\textwidth}
		\begin{tikzpicture}
		\node[ran] (while) at (0,0)  {$\loc_0$};
		\node[ran, right = 2cm of while] (term) {$ \locterm $};
		\node[ran, below = 1 cm of while] (assignx1) {$ \loc_1 $};
		\node[ran, below = 1 cm of assignx1] (assigny) {$ \loc_2 $};
		\node[ran, below = 1 cm of assigny] (while2) {$ \loc_3 $};
		\node[ran, below = 1 cm of while2] (assignx2) {$ \loc_4 $};
		
		\draw[tran] (while) to node[font=\scriptsize,draw, fill=white, 
		rectangle,pos=0.5] {$x<9\land I_{x,y}$} (term);
		\draw[tran] (while) to node[font=\scriptsize,draw, fill=white, 
		rectangle,pos=0.5] {$x\geq 9\land I_{x,y}$} (assignx1);
		\draw[tran] (assignx1) to  node[font=\scriptsize,draw, fill=white, 
		rectangle,pos=0.5] {$y'=y$} (assigny);
		\draw[tran] (assigny) to  node[font=\scriptsize,draw, fill=white, 
		rectangle,pos=0.5] {$y' = 10x\land x'=x$} (while2);
		\draw[tran] (while2) to node[font=\scriptsize,draw, fill=white, 
		rectangle,pos=0.5] {$x\leq y\land I_{x,y}$} (assignx2);
		\node[left = 2cm of assignx2, circle, minimum size = 3mm] (dum) {};
		\draw[tran, rounded corners] (assignx2) -- (dum.east) -- node[font=\scriptsize,draw, fill=white, 
		rectangle,pos=0.5] {$x'=x+1\land y'=y$} (dum.east|-while2.210) -- (while2.210);
		\node[left = 1.7cm of while2, circle, minimum size = 3mm] (dum2) {};
		\draw[tran, rounded corners] (while2) -- (dum2.east) -- node[font=\scriptsize,draw, fill=white, 
		rectangle,pos=0.5] {$x>y\land I_{x,y}$} (dum2.east|-while) -- (while);
		\draw[tran, loop below] (term) to  (term);
		\end{tikzpicture}
	\end{subfigure}
	\begin{subfigure}{0.35\textwidth}
		\begin{tikzpicture}
		\node[ran] (while) at (0,0)  {$\loc_0$};
		\node[ran, right = 2cm of while] (term) {$ \locterm $};
		\node[ran, below = 1 cm of while] (assignx1) {$ \loc_1 $};
		\node[ran, below = 1 cm of assignx1] (assigny) {$ \loc_2 $};
		\node[ran, below = 1 cm of assigny] (while2) {$ \loc_3 $};
		\node[ran, below = 1 cm of while2] (assignx2) {$ \loc_4 $};
		
		\draw[tran] (term) to node[font=\scriptsize,draw, fill=white, 
		rectangle,pos=0.5] {$x<9\land I_{x,y}$} (while);
		\draw[tran] (assignx1) to node[font=\scriptsize,draw, fill=white, 
		rectangle,pos=0.5] {$x\geq 9\land I_{x,y}$} (while);
		\draw[tran] (assigny) to node[font=\scriptsize,draw, fill=white, 
		rectangle,pos=0.5] {$y=y'$} (assignx1);
		\draw[tran] (while2) to node[font=\scriptsize,draw, fill=white, 
		rectangle,pos=0.5] {$y=10x'\land x=x'$} (assigny);
		\draw[tran] (assignx2) to node[font=\scriptsize,draw, fill=white, 
		rectangle,pos=0.5] {$x\leq y\land I_{x,y}$} (while2);
		\node[left = 2.05cm of while2.210, circle, minimum size = 3mm] (dum) {};
		\node[left = 2cm of assignx2, circle, minimum size = 3mm] (dum3) {};
		\draw[tran, rounded corners] (while2.210) -- (dum.east) -- node[font=\scriptsize,draw, fill=white, 
		rectangle,pos=0.5] {$x=x'+1\land y=y'$} (dum3.east|-assignx2) -- (assignx2);
		\node[left = 1.7cm of while, circle, minimum size = 3mm] (dum2) {};
		\draw[tran, rounded corners] (while) -- (dum2.east) -- node[font=\scriptsize,draw, fill=white, 
		rectangle,pos=0.5] {$x>y\land I_{x,y}$} (dum2.east|-while2) -- (while2);
		\draw[tran, loop below] (term) to  (term);
		\end{tikzpicture}
	\end{subfigure}
	\caption{Running example, its associated transition system, and its reversed transition system. $I_{x,y}$ denotes $x'=x\land y'=y$ and is used for readability.}
	\label{fig:running}
\end{figure*}

\noindent{\em Invariants and inductive predicate maps.} Given a transition system $\trsys$, a \textit{ predicate map} is a map $I$ assigning to each location in $\trsys$ a predicate over the program variables. A predicate map naturally defines a set of configurations in $\trsys$ and we will freely interchange between the two notions. A predicate map is \textit{of type-$(c,d)$} if it assigns to each program location a propositional predicate which is a disjunction of $d$ assertions, each being a conjunction of $c$ polynomial inequalities. For a predicate map $I$, we define the complement predicate map $\neg I$ as $(\neg I)(l)=\neg I(l)$ for each location $l$. \\
A predicate map $I$ is said to be an \textit{invariant} if for every reachable configuration $(l,\mathbf{x})$ in $\trsys$, we have $\mathbf{x}\models I(l)$. Intuitively, invariants are over-approximations of the set of reachable configurations in the transition system. A predicate map is \textit{inductive} if it is inductive with respect to every transition $\tau=(l,l',\rho_{\tau})$, i.e.~if for any pair of configurations $(l,\mathbf{x})$ and $(l',\mathbf{x}')$ with $\mathbf{x}\models I(l)$ and $(\mathbf{x},\mathbf{x}')\in \rho_{\tau}$, we also have $\mathbf{x}'\models I(l')$.

\smallskip\noindent{\em Termination problem.} Given a program and its transition system $\trsys$, we say that a run reaching $\locterm$ is \textit{terminating}. The program is said to be \textit{terminating} if every run in $\trsys$ is terminating. Otherwise it is said to be \textit{non-terminating}. One witness to non-termination can be a configuration that is reachable but from which there are no terminating executions. We call such configuration \textit{diverging}.

\begin{example}\label{ex:nonterm}
Consider again the running example in Fig.~\ref{fig:running}. For any initial configuration with $x\geq 9$, executions that always assign $x:=9$ when passing the non-deterministic assignment are non-terminating. On the other hand, the execution that assigns $x:=0$ in the non-deterministic assignment enters the outer loop only once and then terminates. Thus, no initial configuration is diverging. One can similarly check that other configurations are also not diverging.
\end{example}

\section{Transition system reversal}\label{sec:reversalcertificate}\label{sec:reverse}

We now show that it is possible to ''reverse'' a transition system by reversing each of its transitions. This construction is the core concept of our approach to proving non-termination, since configurations in the program from which $\locterm$ is reachable will be precisely those configurations which can be reached from $\locterm$ in the reversed transition system. We then present a sound and complete certificate for non-termination based on this construction. 

\begin{definition}[Reversed transition system]\label{def:reverse}
	Given a transition system $\trsys=(L,\vars,\locinit,\Theta_{init},\transitions)$ and a transition $\tau=(l,l',\rho_{\tau})\in\, \transitions$, let
	\begin{equation*}
	\rho'_{\tau}=\{(\mathbf{x}',\mathbf{x})\mid (\mathbf{x},\mathbf{x}')\in\rho_{\tau}\}.
	\end{equation*}
	If $\rho_{\tau}$ is given by an assertion over $\vars\cup\vars'$, $\rho'_{\tau}$ is obtained from $\rho_{\tau}$ by replacing each unprimed variable in the defining assertion for $\rho_{\tau}$ with its primed counterpart, and vice-versa. Then for an assertion $\Theta$, we define the {\em reversed transition system} of $\trsys$ with initial variable valuations $\Theta$ as a tuple $\trsys^{r,\Theta}=(L,\vars,\locterm,\Theta,\transitions^r)$, where $\transitions^r=\{(l',l,\rho'_{\tau})\mid (l,l',\rho_{\tau})\in\,\transitions\}$.
\end{definition}

\noindent Note that this construction satisfies Definition~\ref{def:cfg} and thus yields another transition system. All notions that were defined before (e.g.~configuration, finite path, etc.) are defined analogously for the reversed transition systems.

\begin{example}
	Fig.~\ref{fig:running} right shows the reversed transition system $\trsys^{r,\Theta}$ of the program in Fig.~\ref{fig:running}. Note that for every transition $\tau$ in $\trsys$ for which $\rho_{\tau}$ is given by a conjunction of an assertion over unprimed program variables and $x'=x\land y'=y$, after reversing we obtain the conjunction of the same assertion just now over primed variables and $x'=x\land y'=y$. Hence, for such $\tau$ the transition relation is invariant under reversing. For example, a transition from $l_0$ to $l_1$ in $\trsys$ has transition relation $x\geq 9 \land x'=x\land y'=y$ so the reversed transition has transition relation $x'\geq 9\land x=x'\land y=y'$. As $x'=x$, this is the same relation as prior to reversal.
\end{example}

\noindent The following lemma is the key property of this construction.

\begin{lemma}[Key property of reversed transition systems]\label{lemma:key}
	Let $\trsys$ be a transition system, $\Theta$ an assertion and $\trsys^{r,\Theta}$ the reversed transition system of $\trsys$ with initial variable valuations $\Theta$. Let $\mathbf{c}$ and $\mathbf{c}'$ be two configurations. Then $\mathbf{c}'$ is reachable from $\mathbf{c}$ in $\trsys$ if and only if $\mathbf{c}$ is reachable from $\mathbf{c}'$ in $\trsys^{r,\Theta}$.
\end{lemma}

\begin{proof}
	We prove that if $\mathbf{c}'$ is reachable from $\mathbf{c}$ in $\trsys$ then $\mathbf{c}$ is reachable from $\mathbf{c}'$ in $\trsys^{r,\Theta}$, the other direction follows analogously.
	Suppose that $\mathbf{c}=(l_0,\mathbf{x}_0),(l_1,\mathbf{x}_1),\dots,(l_k,\mathbf{x}_k)=\mathbf{c}'$ is a path from $\mathbf{c}$ to $\mathbf{c}'$ in $\trsys$. Then for each $0\leq i<k$ there is a transition $\tau_i=(l_i,l_{i+1},\rho_{\tau_i})$ in $\trsys$ for which $(\mathbf{x}_i,\mathbf{x}_{i+1})\in \rho_{\tau_i}$. But then $(\mathbf{x}_{i+1},\mathbf{x}_i)\in \rho'_{\tau_i}$ and $\tau_i^r=(l',l,\rho'_{\tau_i})$, hence $(l_i,\mathbf{x}_i)$ is a successor of $(l_{i+1},\mathbf{x}_{i+1})$ in $\trsys^{r,\Theta}$. Thus \[ \mathbf{c'}=(l_k,\mathbf{x}_k),(l_{k-1},\mathbf{x}_{k-1}),\dots,(l_0,\mathbf{x}_0)=\mathbf{c}\]
	is a finite path in $\trsys^{r,\Theta}$, proving the claim.
\end{proof}

\noindent{\em Backward Invariants.} Lemma~\ref{lemma:key} implies that generating invariants for the reversed transition system $\trsys^{r,\Theta}$ provides a way to over-approximate the set of configurations in $\trsys$ from which some configuration in the set $\{(\locterm,\mathbf{x})\mid \mathbf{x}\models\Theta\}$ is reachable. This motivates the notion of a \textit{backward invariant}, which will be important in what follows.

\begin{definition}[Backward invariant]\label{def:backward}
	For a transition system $\trsys$ and an assertion $\Theta$, we say that the predicate map $\BI$ is a {\em backward invariant in $\trsys^{r,\Theta}$} if it is an invariant in $\trsys^{r,\Theta}$. The word backward is used to emphasize that we are working in the reversed transition system.
\end{definition}

\noindent We conclude this section with a theorem illustrating the behavior of inductive predicate maps under program reversal.

\begin{theorem}\label{lemma:inductiverreverse}
	Let $\trsys$ be a transition system, $\Theta$ an assertion, $I$ a predicate map and $\trsys^{r,\Theta}$ the reversed transition system. Then $I$ is inductive in $\trsys$ if and only if $\neg I$ is inductive in $\trsys^{r,\Theta}$.
\end{theorem}

\begin{proof}
	We show that $I$ being inductive in $\trsys$ implies that $\neg I$ is inductive in $\trsys^{r,\Theta}$. The other direction of the lemma follows analogously. \\
	Let $\tau^r=(l',l,\rho'_{\tau})$ be a transition in $\trsys^{r,\Theta}$ obtained by reversing $\tau=(l,l',\rho_{\tau})$ in $\trsys$. Assume that $\mathbf{x}'\in\neg I(l')$. To show inductiveness of $\neg I$ in the reversed transition system, we take a successor $(l,\mathbf{x})$ of $(l’,\mathbf{x}’)$ in the reversed transition system with $(\mathbf{x}',\mathbf{x})\in\rho'_{\tau}$, and we need to show that $\mathbf{x} \in \neg I(l)$. By definition of the reversed transition we have $(\mathbf{x},\mathbf{x}')\in\rho_{\tau}$. So, if on the contrary we had $\mathbf{x}\in I(l)$, inductiveness of $I$ in $\trsys$ would imply  that $\mathbf{x}'\in I(l')$. This would contradict the assumption that $\mathbf{x}'\in\neg I(l')$. Thus, we must have $\mathbf{x}\in \neg I(l)$, and $\neg I$ is inductive in $\trsys^{r,\Theta}$.
\end{proof}

\section{Sound and Complete Certificate for Non-termination}\label{sec:certificate}

Lemma~\ref{lemma:key} indicates that reversed transition systems are relevant for the termination problem, as they provide means to describe configurations from which the terminal location can be reached. We now introduce the \textit{$\BI$-certificate for non-termination}, based on the reversed transition systems and backward invariants. We show that it is both \textit{sound} and \textit{complete} for proving non-termination and hence characterizes it (i.e. a program is non-terminating if and only if it admits the certificate). This is done by establishing a connection to recurrence sets~\cite{GuptaHMRX08,ChenCFNO14}, a notion which provides a necessary and sufficient condition for a program to be non-terminating.

\smallskip\noindent{\em Recurrence set.} A \textit{recurrence set}~\cite{GuptaHMRX08} in a transition system $\trsys$ is a non-empty set of configurations $\mathcal{G}$ which (1)~contains some configuration reachable in $\trsys$, (2)~every configuration in $\mathcal{G}$ has at least one successor in $\mathcal{G}$, and (3)~contains no terminal configurations. The last condition was not present in~\cite{GuptaHMRX08} and we add it to account for the terminal location and the self-loop at it, but the definitions are easily seen to be equivalent. In~\cite{GuptaHMRX08}, it is shown that a program is non-terminating if and only if its transition system contains a recurrence set. The work in~\cite{ChenCFNO14} notes that one may without loss of generality restrict attention to recurrence sets which contain some initial configuration (which they call~\textit{open recurrence sets}). Indeed, to every recurrence set one can add configurations from some finite path reaching it to obtain an open recurrence set, and there is at least one such path since each recurrence set contains a reachable configuration.

\smallskip\noindent{\em Closed recurrence set.} A \textit{closed recurrence set}~\cite{ChenCFNO14} is an open recurrence set $\mathcal{C}$ with the additional property of being inductive, i.e.~for every configuration in $\mathcal{C}$ \textit{each} of its successors is also contained in $\mathcal{C}$. The work~\cite[Theorems 1 and 2]{ChenCFNO14} shows that closed recurrence sets can be used to define a sound and complete certificate for non-termination, which we describe next. Call $U=(L,\vars,\locinit,\Theta_{init},\transitions_U)$ an \textit{under-approximation} of $\trsys=(\locs,\vars,\locinit,\Theta_{init},\transitions)$ if for every transition $(l,l',\rho^u_{\tau})\in\,\transitions_U$ there exists $(l,l',\rho_{\tau})\in\,\transitions$ with $\rho^u_{\tau}\subseteq \rho_{\tau}$. Then $\trsys$ contains an open recurrence set if and only if there is an under-approximation $U$ of $\trsys$ and a closed recurrence set in $U$. 

\smallskip\noindent{\em Proper under-approximations.} We introduce a notion of proper under-approximation. An under-approximation $U$ of $\trsys$ is {\em proper} if every configuration which has a successor in $\trsys$ also has at least one successor in $U$. This is a new concept and restricts general under-approximations, but it will be relevant in defining the $\BI$-certificate for non-termination and establishing its soundness and completeness. The next lemma is technical and shows that closed recurrence sets in proper under-approximations are sound and complete for proving non-termination, its proof can be found in Appendix~\ref{app:certificate}.

\begin{lemma}\label{lemma:properun}
	Let $P$ be a non-terminating program and $\trsys$ its transition system. Then there exist a proper under-approxima\-tion $U$ of $\trsys$ and a closed recurrence set $\mathcal{C}$ in $U$.
\end{lemma}

\smallskip\noindent{\em $\BI$-certificate for non-termination.} We introduce and explain how backward invariants in combination with proper under-approximations can be used to characterize non-termination. 
Suppose $P$ is a program we want to show is non-terminating, and $\trsys$ is its transition system. Let $\mathit{Reach}_{\trsys}(\locterm)$ be the set of variable valuations of all reachable terminal configurations in $\trsys$. A \textit{$\BI$-certificate for non-termination} will consist of an ordered triple $(U,\BI,\Theta)$ of a proper under-approximation $U$ of $\trsys$, a predicate map $\BI$ and an assertion $\Theta$ such that
\begin{compactitem}
	\item $\Theta\supseteq\mathit{Reach}_{\trsys}(\locterm)$;
	\item $\BI$ is an inductive backward invariant in $U^{r,\Theta}$;
	\item $\BI$ is not an invariant in $\trsys$.
\end{compactitem}

\begin{theorem}[Soundness of our certificate]\label{thm:certsoundness}
	Let $P$ be a program and $\trsys$ its transition system. If there exists a $\BI$-certificate $(U,\BI,\Theta)$ in $\trsys$, then $P$ is non-terminating.
\end{theorem}

\begin{proof}[Proof sketch]
	As $\BI$ is not an invariant in $\trsys$, its complement $\neg \BI$ contains a reachable configuration $\mathbf{c}$. On the other hand, $\BI$ is inductive in $U^{r,\Theta}$ so by Theorem~\ref{lemma:inductiverreverse} $\neg \BI$ is inductive in $U$. Since $U$ is proper (and since in transition systems induced by programs every configuration has a successor), one may take a finite path reaching $\mathbf{c}$ and inductively keep picking successors in $U$ from $\mathbf{c}$, obtaining an execution whose all but finitely many configurations are in $\neg \BI$. By the definition of $\Theta$ and since $\BI$ is an invariant for $U^{r,\Theta}$, $\neg \BI$ contains no reachable terminal configuration hence this execution is non-terminating. Details can be found in Appendix~\ref{app:certificate}.
\end{proof}

\begin{example}\label{ex:certifiex}
	Consider again the running example and its transition system $\trsys$ presented in Fig.~\ref{fig:running}. Let $U$ be the under-approximation of $\trsys$ defined by restricting the transition relation of the non-deterministic assignment $x:=\textbf{ndet}()$ as $\rho^U_{\tau}=\{(x,y,x',y')\mid x'=9, y'=y\}$. Intuitively, $U$ is a transition system of the program obtained by replacing the non-deterministic assignment in $P$ with $x:=9$. Define a predicate map $\BI$ as
	\begin{equation*}
	\BI(l) = \begin{cases}
	(1\geq 0) &\text{if $l=\locterm$}\\
	(x\leq 8) &\text{if $l\in\{l_0,l_2,l_3,l_4$\}}\\
	(-1\geq 0) &\text{if $l=l_1$},
	\end{cases}
	\end{equation*}
	i.e.~$\BI(l_1)$ is empty, and let $\Theta=\mathbb{Z}^2$. $U^{r,\Theta}$ can be obtained from $\trsys^{r,\Theta}$ by replacing the transition relation from $l_2$ to $l_1$ with $x=9\land y=y'$ in Fig.~\ref{fig:running} right. Then $U$ is proper, and $BI$ is an inductive backward invariant for $U^{r,\Theta}$ since no transition can increase $x$. On the other hand, $(l_0,9,0)$ is reachable in $\trsys$ but not contained in $\BI$, thus $\BI$ is not an invariant in $\trsys$. Hence $(U,\BI,\Theta)$ is a $\BI$-certificate for non-termination and the program is non-terminating.
\end{example}

\noindent By making a connection to closed recurrence sets, the following theorem shows that backward invariants in combination with proper under-approximations of $\trsys$ also provide a \textit{complete characterization} of non-termination.

\begin{theorem}[Complete characterization of non-termination]\label{thm:completechar}
	Let $P$ be a non-terminating program with transition system $\trsys$. Then $\trsys$ admits a proper under-approximation $U$ and a predicate map $\BI$ such that $\BI$ is an inductive backward invariant in the reversed transition system $U^{r,\mathbb{Z}^{|\vars|}}$, but not an invariant in $\trsys$.
\end{theorem}

\begin{proof}[Proof sketch]
	Since $P$ is non-terminating, from Lemma~\ref{lemma:properun} we know that $\trsys$ admits a proper under-approximation $U$ and a closed recurrence set $C$ in $U$. For each location $l$ in $\trsys$, let $C(l)=\{\mathbf{x}\mid (l,\mathbf{x})\in C\}$. Define the predicate map $\BI$ as $\BI(l)=\neg C(l)$ for each $l$. Then, using Theorem~\ref{lemma:inductiverreverse} one can show that $U$ and $\BI$ satisfy the conditions of the theorem. For details, see Appendix~\ref{app:certificate}.
\end{proof}

\begin{rmk}[Connection to the \emph{pre}-operator]
\label{rmk:pre}
There is a certain similarity between reversal of an individual transition and application of the \emph{pre}-operator, the latter being a well known concept in program analysis. However, in our approach we introduce reversed transition systems which are obtained by reversing {\em all transitions} (hence the name ``program reversal''). This allows us using black-box invariant generation techniques as a \emph{one-shot} method of computing sets from which a terminal location can be reached, as presented in the next section. This is in contrast to approaches which rely on an iterative application of the \emph{pre}-operator.
\end{rmk}

\section{Algorithm for Proving Non-termination}\label{sec:nonterm}

We now present our algorithm for proving non-termination based on program reversing and $\BI$-certificates introduced in Section~\ref{sec:certificate}. It uses a black box constraint solving-based method for generating (possibly disjunctive) inductive invariants, as in~\cite{ColonSS03,GulwaniSV08,KincaidBBR17,KincaidCBR18,HrushovskiOP018,Rodriguez-CarbonellK04,Rodriguez-CarbonellK07,ChatterjeeFGG19}. This is a classical approach to invariant generation and it fixes a template for the invariant (i.e.~a type-$(c,d)$ propositional predicate map as well as an upper bound $D$ on the degree of polynomials, where $c$, $d$ and $D$ are provided by the user), introduces a fresh variable for each template coefficient, and encodes invariance and inductiveness conditions as existentially quantified constraints on template coefficient variables. The obtained system is then solved and any solution yields an inductive invariant. Moreover, the method is relatively complete~\cite{Rodriguez-CarbonellK04,Rodriguez-CarbonellK07,ChatterjeeFGG19} in the sense that every inductive invariant of the fixed template and maximal polynomial degree is a solution to the system of constraints. Efficient practical approaches to polynomial inductive invariant generation have been presented in~\cite{KincaidBBR17,KincaidCBR18}.


We first introduce \textit{resolution of non-determinism} which induces a type of proper under-approximations of the program's transition system of the form that allows searching for them via constraint solving. We then proceed to our main algorithm. In what follows, $P$ will denote a program with polynomial arithmetic and $\trsys=(\locs,\vars,\locinit,\Theta_{init},\transitions)$ will be its transition system.

\subsection{Resolution of non-determinism}\label{sec:resnondet}

As we saw in Example~\ref{ex:nonterm}, there may exist non-diverging program configurations which become diverging when supports of non-deterministic assignments are restricted to suitably chosen subsets. Here we define one such class of restrictions which ''resolves'' each non-deterministic assignment by replacing it with a polynomial expression over program variables. Such resolution ensures that the resulting under-approximation of the program's transition relation is proper. Let $T_{\NA}\subseteq\, \transitions$ be the set of transitions corresponding to non-deterministic assignments in $P$. 

\begin{definition}[Resolution of non-determinism]\label{def:restriction}
	A {\em resolution of non-determinism} for $\trsys$ is a map $R^{\NA}$ which to each $\tau\in T_{\NA}$ assigns a polynomial expression $R^{\NA}(\tau)$ over program variables. It naturally defines a {\em restricted transition system} $\trsys_{R^{\NA}}$ which is obtained from $\trsys$ by letting the transition relation of $\tau\in T_{\NA}$ corresponding to an assignment $x:=\textbf{ndet}()$ be
	\begin{equation*}
	\rho^{R^{\NA}}_{\tau}(\mathbf{x},\mathbf{x}') := (x'=R^{\NA}(\tau)(\mathbf{x}))\land \bigwedge_{y\in \vars\backslash\{x\}}y'=y.
	\end{equation*}
\end{definition}

\noindent Note that $\trsys_{R^{\NA}}$ is a proper under-approximation of $\trsys$. If there exists a resolution of non-determinism $R^{\NA}$ and a configuration $\mathbf{c}$ which is reachable in $\trsys$ but from which no execution in $\trsys_{R^{\NA}}$ terminates, then any such execution is non-terminating in $\trsys$ as well. We say that any such configuration $\mathbf{c}$ is {\em diverging with respect to (w.r.t.)~$R^{\NA}$}.

\begin{example}
	Looking back at the program in Figure~\ref{fig:running}, define a resolution of non-determinism $R^{\NA}$ to assign constant expression $9$ to the non-deterministic assignment $x:=\textbf{ndet}()$. Then every initial configuration with $x\geq 9$ becomes diverging w.r.t.~$R^{\NA}$.
\end{example}


\subsection{Algorithm}\label{sec:overview}

\noindent{\em Main idea.} To prove non-termination, our algorithm uses a constraint solving approach to find a $\BI$-certificate. It searches for a resolution of non-determinism $R^{\NA}$, a propositional predicate map $\BI$ and an assertion $\Theta$ such that:
\begin{compactenum}
	\item $\Theta\supseteq \mathit{Reach}_{\trsys}(\locterm)$ (recall that $\mathit{Reach}_{\trsys}(\locterm)$ is the set of variable valuations of all reachable terminal configurations in $\trsys$);
	\item $\BI$ is an inductive backward invariant for the reversed transition system $\trsys^{r,\Theta}_{R^{\NA}}$;
	\item $\BI$ is not an invariant for $\trsys$.
\end{compactenum}

\smallskip\noindent{\em Need for inductive invariants and safety checking.} Using the aforementioned black box invariant generation, our algorithm encodes the conditions on $R^{\NA}$, $\BI$, and $ \Theta $ as polynomial constraints and then solves them. However, the method is only able to generate \textit{inductive} invariants, which is to say that encoding ''$\BI$ is not an invariant for $\trsys\,$'' is not possible. Instead, we modify the third requirement on $\BI$ above to get:
\begin{compactenum}
	\item $\Theta\supseteq \mathit{Reach}_{\trsys}(\locterm)$;
	\item $\BI$ is an inductive backward invariant for $\trsys^{r,\Theta}_{R^{\NA}}$;
	\item $\BI$ is not an inductive invariant for $\trsys$.
\end{compactenum}
The third requirement does not guarantee that we get a proper $\BI$-certificate. However it guides invariant generation to search for $\BI$ which is less likely to be an invariant for $\trsys$. It follows that the algorithm needs to do additional work to ensure that the triple $(R^{\NA},\BI, \Theta) $ is a $ \BI $-certificate.

\smallskip\noindent{\em Splitting the algorithm into two checks.} The predicate map $\BI$ is not an inductive invariant for $\trsys$ if and only if it has one of the following properties: either it does not contain some initial configuration or is not inductive with respect to some transition in $\trsys$. For each of these two properties, we can separately compute $ \BI $ satisfying it and the properties~(1) and~(2) above, followed by a check whether the computed $ \BI $ indeed proves non-termination. We refer to these two independent computations as two checks of our algorithm:
\begin{compactitem}
	\item {\em Check 1} - the algorithm checks if there exist $R^{\NA}$, $\BI$ and $\Theta$ as above so that $\BI$ does not contain some initial configuration and conditions~(1) and~(2) are satisfied. By Theorem~\ref{lemma:inductiverreverse}, $\BI$ is inductive for $\trsys^{r,\Theta}_{R^{\NA}}$ if and only if the complement $\neg\BI$ is an inductive predicate map for $\trsys_{R^{\NA}}$. Moreover, since $\neg\BI$ contains an initial configuration there is no need for an additional reachability check to conclude that $\BI$ is not an invariant for $\trsys$. Hence by fixing $\Theta=\mathbb{Z}^{|\vars|}$, to prove non-termination it suffices to check if there exist a resolution of non-determinism $R^{\NA}$, a predicate map $I$ and an initial configuration $\mathbf{c}$ in $\trsys$ such that $I$ contains $\mathbf{c}$, $I$ is inductive for $\trsys_{R^{\NA}}$ and $I(\locterm)=\emptyset$.
	
	\item {\em Check 2} - the algorithm checks if there exist $R^{\NA}$, $\Theta$ and $\BI$ as above so that $\BI$ is not inductive in $\trsys$ and conditions~(1) and~(2) are satisfied. If a solution is found, the algorithm still needs to find a configuration in $\neg\BI$ which is reachable in $\trsys$, via a call to a safety prover.
\end{compactitem}

\begin{algorithm}[t]
	\SetKwInOut{Input}{input}\SetKwInOut{Output}{output}
	\DontPrintSemicolon
	
	\Input{A program $P$, its transition system $\trsys$, predicate map template size $(c,d)$, maximal polynomial degree $D$.}
	\Output{Proof of non-termination if found, otherwise ''Unknown''}
	\medskip
	set a template for each polynomial defined by resolution of non-determinism $R^{NA}$\\
	construct restricted transition system $\trsys_{R^{NA}}$\\
	set templates for configuration $\mathbf{c}$ and for an invariant $I$ of type-$(c,d)$\\
	encode $\Phi_1=\phi_{\mathbf{c}}\land\phi_{I,R^{\NA}}$\\
	\lIf{$\Phi_1$ feasible}{\Return Non-termination}
	\Else{
		set templates for invariant $\tilde{I}$ of type-$(c,1)$ and for a backward invariant $\BI$ of type-$(c,d)$\\
		construct reversed transition system $\trsys^{r,\tilde{I}(\locterm)}_{R^{NA}}$\\
		\lForEach{$\tau\in\,\transitions$}{
			set templates for $\mathbf{x}_{\tau}$, $\mathbf{x}'_{\tau}$}
		encode $\Phi_{2}=\phi_{\tilde{I}}\land\phi_{\BI,R^{\NA}}\land\bigvee_{\tau\in\,\transitions}\phi_{\tau}$\\
		\If{$\Phi_2$ feasible}{
			\lIf{$\exists$ $(l,\mathbf{x})$ Reachable in $\trsys$ with $\mathbf{x}\models\neg \BI(l)$}{\Return Non-termination}
			\lElse{\Return Unknown}}
		\lElse{\Return Unknown}
	}
	\caption{Proving non-termination}
	\label{algo:nonterm}
\end{algorithm}

\smallskip\noindent{\em Algorithm summary.} As noted at the beginning of Section~\ref{sec:nonterm}, the invariant generation method first needs to fix a template for the propositional predicate map and the maximal polynomial degree. Thus our algorithm is parametrized by $c$ and $d$ which are bounds on the template size of propositional predicate maps ($d$ being the maximal number of disjunctive clauses and $c$ being the maximal number of conjunctions in each clause), and by an upper bound $D$ on polynomial degrees.
The algorithm consists of two checks, which can be executed either sequentially or in parallel:

	{\em Check 1} - the algorithm checks if there exist a resolution of non-determinism $R^{\NA}$, a predicate map $I$ and an initial configuration $\mathbf{c}$ such that (1)~$I$ is an inductive invariant in $\trsys_{R^{\NA}}$ for the single initial configuration $\mathbf{c}$, and (2)~$I(\locterm)=\emptyset$. To do this, we fix a template for each of $R^{\NA}$, $I$ and $\mathbf{c}$, and encode these properties as polynomial constraints: 
	\begin{compactitem}
		\item For each transition $\tau$ in $T_{\NA}$, fix a template for a polynomial $R^{\NA}(\tau)$ over program variables of degree at most $D$. That is, introduce a fresh template variable for each coefficient of such a polynomial.
		\item Introduce fresh variables $c_1,c_2,\dots,c_{|\vars|}$ defining the variable valuation of $\mathbf{c}$. Then substitute these variables into the assertion $\Theta_{init}$ specifying initial configurations in $\trsys$ to obtain the constraint $\phi_{\mathbf{c}}$ for $\mathbf{c}$ being an initial configuration.
		\item Fix a template for the propositional predicate map $I$ of type-$(c,d)$ and maximal polynomial degree $D$. The fact that $I$ is an inductive invariant for $\trsys_{R^{\NA}}$ with the single initial configuration $\mathbf{c}$ and $I(\locterm)=\emptyset$ is encoded by the invariant generation method (e.g.~\cite{ColonSS03,Rodriguez-CarbonellK04}) into a constraint $\phi_{I,R^{\NA}}$.
	\end{compactitem}
	The algorithm then tries to solve $\Phi_1=\phi_{\mathbf{c}}\land\phi_{I,R^{\NA}}$ using an off-the-shelf SMT solver. If a solution is found, $\mathbf{c}$ is an initial diverging configuration w.r.t.~$\trsys_{R^{\NA}}$, so the algorithm reports non-termination.
	
 {\em Check 2} - the algorithm checks if there exist a resolution of non-determinism $R^{\NA}$, an assertion $\Theta$, a predicate map $\BI$ and a transition $\tau\in T_{\NA}$ such that (1)~$\Theta\supseteq \mathit{Reach}_{\trsys}(\locterm)$, (2)~$\BI$ is an inductive backward invariant for $\trsys^{r,\Theta}_{R^{\NA}}$, and (3)~$\BI$ is not inductive w.r.t.~$\tau$ in $\trsys$. 
	To encode $\Theta\supseteq\mathit{Reach}_{\trsys}(\locterm)$, we introduce another propositional predicate map $\tilde{I}$ (purely conjunctive for the sake of efficiency), and impose a requirement on it to be an inductive invariant for $\trsys$. We may then define the initial variable valuations for $\trsys^{r,\Theta}_{R^{\NA}}$ as $\Theta=\tilde{I}(\locterm)$. The algorithm introduces fresh template variables for $R^{\NA}$, $\tilde{I}$ and $\BI$, as well as for a pair of variable valuations $\mathbf{x}_{\tau}$ and $\mathbf{x}'_{\tau}$ for each transition $\tau=(l,l',\rho_{\tau})$ in $\trsys$ and imposes the following constraints:
	\begin{compactitem}
		\item For each transition $\tau$ in $T_{\NA}$, fix a template for a polynomial expression $R^{\NA}(\tau)$ of degree at most $D$ over program variables.
		\item Fix a template for the propositional predicate map $\tilde{I}$ of type-$(c,1)$ (as explained above, for efficiency reasons we make $\tilde{I}$ conjunctive) and impose a constraint $\phi_{\tilde{I}}$ that $\tilde{I}$ is an inductive invariant for $\trsys$.
		\item Fix a template for the propositional predicate map $BI$ of type-$(c,d)$ and impose a constraint $\phi_{\BI,R^{\NA}}$ that $\BI$ is an inductive backward invariant for $\trsys^{r,\tilde{I}(\locterm)}_{R^{\NA}}$.
		\item For each transition $\tau$ in $\trsys$, the constraint $\phi_{\tau}$ encodes non-inductiveness of $\BI$ with respect to $\tau$ in $\trsys$:
		\begin{equation*}
		\mathbf{x},\mathbf{x}'\models \BI(l)\land \rho_{\tau}\land \neg \BI(l').
		\end{equation*}
	\end{compactitem}
	The algorithm then solves $\Phi_{2}=\phi_{\tilde{I}}\land\phi_{\BI,R^{\NA}}\land\bigvee_{\tau\in\,\transitions}\phi_{\tau}$
	by using an SMT-solver. If a solution is found, the algorithm uses an off-the-shelf safety prover to check if there exists a configuration in $\neg\BI$ reachable in $\trsys$. Such configuration is then diverging w.r.t.~$\trsys_{R^{\NA}}$, so we report non-termination.

The pseudocode for our algorithm is shown in Algorithm~\ref{algo:nonterm}. The following theorem proves soundness of our algorithm, and its proof can be found in Appendix~\ref{app:certificate}.

\begin{theorem}[Soundness]\label{thm:soundness}
	If Algorithm~\ref{algo:nonterm} outputs ''Non-termi\-nation'' for some input program $P$, then $P$ is non-terminating.
\end{theorem}

\begin{rmk}[Algorithm termination]
	Our algorithm might not always terminate because either the employed SMT-solver or the safety prover might diverge. Thus, in practice one needs to impose a timeout in order to ensure algorithm termination.
\end{rmk}

\subsection{Demonstration on Examples}\label{sec:algex}

We demonstrate our algorithm on two examples illustrating the key aspects. 
In Appendix~\ref{app:aperiodic}, we present an example demonstrating an application of our method on program whose all non-terminating traces are aperiodic.

\begin{example}\label{ex:runningnonterm}
	Consider again our running example in Fig.~\ref{fig:running}. We demonstrate that Check 1 of our algorithm can prove that it is non-terminating. Define the resolution of non-deter\-minism $R^{\NA}$ to assign a constant expression $9$ to the non-deterministic assignment, an initial configuration $\mathbf{c}=(\locinit,9,0)$, and a propositional predicate map $I$ as $I(\loc)=(x\geq 9)$ for $\loc\neq \locterm$ and $I(\locterm)=\emptyset$. Then $I$ is an inductive invariant for $\trsys_{R^{\NA}}$ with the initial configuration $\mathbf{c}$. Thus the system of polynomial constraints constructed by Check 1 is feasible, proving that this program is non-terminating.
\end{example}

\lstset{language=affprob}
\lstset{tabsize=2}
\newsavebox{\exappp}
\begin{lrbox}{\exappp}
	\begin{lstlisting}[mathescape]
	$n := 0,\, b := 0,\, u:=0$
	$l_0$:	while $b == 0$ and $n\leq 99$ do
	$l_1$:			$u := \textbf{ndet}()$
	$l_2$:			if $u\leq -1$ then 
	$l_3$:					$b:=-1$
				else if $u==0$ then 
	$l_4$:					$b:=0$
	$l_5$:			else $b:=1$ fi
	$l_6$:			$n := n + 1$
	$l_7$:			if $n \geq 100$ and $b \geq 1$ then
	$l_8$:					while true do
	$l_9$:						skip
	   od fi od
	\end{lstlisting}
\end{lrbox}

\begin{figure}[t]
	\centering
	\usebox{\exappp}
	\caption{An example of a program without an initial diverging configuration with respect to any resolution of non-determinism that uses polynomials of degree less than $100$, but for which Check 2 proves non-termination. }
	\label{fig:noninitialterm}
\end{figure}

\begin{example}\label{ex:noninitialterm}
	Consider the program in Fig.~\ref{fig:noninitialterm}. Its initial variable valuation is given by the assertion $(n=0\land b=0\land u=0)$, and a program execution is terminating so long as it does not assign $0$ to $u$ in the first $99$ iterations of the outer loop, and then at least $1$ in the $100$-th iteration. Thus, if the initial configuration was diverging with respect to a resolution of non-determinism which resolves the non-deterministic assignment of $u$ by a polynomial $p(n,b,u)$, this polynomial would need to satisfy $p(n,0,0)=0$ for $n=0,1,\dots,98$ and $p(99,0,0)\geq 1$. Hence, the degree of $p$ would have to be at least $100$, and this program has no initial diverging configuration with respect to any resolution of non-determinism that is feasible to compute by using the Check 1 of our algorithm.
	
	We now show that Check 2 can prove non-termination of this program using only polynomials of degree $0$, i.e.~constant polynomials. Define $R^{\NA}$, $\Theta$, $\BI$ and $\tau$ as follows:
	\begin{compactitem}
		\item $R^{\NA}$ assigns constant expression $1$ to the assignment of $u$ at $\loc_1$;
		\item $\tilde{I}(\loc)=(0\leq n\leq 100)$ for each location $\loc$;
		\item $\BI$ is a propositional predicate map defined via
		\begin{equation*}
		\BI(\loc) = \begin{cases}
		(0\leq n\leq 100) &\text{if $\loc=\locterm$}\\
		(n\leq 100) &\text{if $\loc=\loc_0$}\\
		(n\leq 99) \lor (n=100\land b\leq 0) &\text{if $\loc=\loc_7$}\\
		(n\leq 98) \lor (n=99\land b\leq 0) &\text{if $\loc=\loc_6$}\\
		(n\leq 98) &\text{if $\loc\in\{\loc_1,\loc_5\}$}\\
		(n\leq 99) &\text{if $\loc\in\{\loc_3,\loc_4\}$}\\
		(n\leq 98) \lor (n=99\land u\leq 0) &\text{if $\loc=\loc_2$}\\
		(1\leq 0) &\text{if $\loc\in\{l_8,l_9\}$};
		\end{cases}
		\end{equation*}
		\item $\tau$ is the transition from $\loc_0$ to $\loc_1$.
	\end{compactitem}
	To show that these $R^{\NA}$, $\tilde{I}$, $\BI$ and $\tau$ satisfy each condition in Check 2, we note that:
	\begin{compactenum}[(1)]
		\item The set of variable valuations reachable in the program upon termination is $(n,b)\in\{(n,b)\mid 1\leq n\leq 99 \land b!=0\} \cup \{(100,b)\mid b\leq 0\}$, thus $\Theta=\tilde{I}(\locterm)$ contains it;
		\item $\BI$ is an inductive backward invariant for $\trsys^{r,\tilde{I}(\locterm)}_{R^{\NA}}$ (which can be checked by inspection of the reversed transition system in Appendix~\ref{app:reversed});
		\item $\BI$ is not inductive w.r.t.~$\tau$ in $\trsys$, since $(99,0,0)\in \BI(\loc_0)$ but the variable valuation $(99,0,0)$ obtained by executing $\tau$ in $\trsys$ is not contained in $\BI(\loc_1)$.
	\end{compactenum}
	Thus, these $R^{\NA}$, $\tilde{I}$, $\BI$ and $\tau$ present a solution to the system of constraints defined by Check 2. Since the configuration $(\loc_1,99,0,0)$ is reachable in this program by assigning $u:=0$ in the first $99$ iterations of the outer loop, but $(99,0,0)\not\in \BI(\loc_1)$, the safety prover will be able to show that a configuration in $\neg \BI$ is reachable. Hence our algorithm is able to prove non-termination.
\end{example}

\subsection{Relative Completeness}\label{sec:completeness}

At the beginning of Section~\ref{sec:nonterm} we noted that constraint solving-based inductive invariant generation is relatively complete~\cite{ColonSS03,GulwaniSV08,Rodriguez-CarbonellK04,Rodriguez-CarbonellK07,ChatterjeeFGG19}, in the sense that whenever there is an inductive invariant representable using the given template, the algorithm will find such an invariant. This means that our algorithm is also relatively complete in checking whether the program satisfies properties encoded as polynomial constraints in Check 1 and Check 2.
Since successful Check~1 does not require a subsequent call to a safety prover, it provides to the best of our knowledge the first \textit{relatively complete algorithm} for proving non-termination of programs with polynomial integer arithmetic and non-determinism.

\begin{theorem}[Relative completeness]\label{thm:completeness}
	Let $P$ be a program with polynomial integer arithmetic and $\trsys$ its transition system. Suppose that $\trsys$ admits a proper under-approximation $U$ which restricts each non-deterministic assignment to a polynomial assignment, and a propositional predicate map $C$ which is a closed recurrence set in $U$. Then for sufficiently high values of parameters $c$, $d$ and $D$ bounding the template size for invariants and the maximal polynomial degree, our algorithm proves non-termination of the program $P$.
\end{theorem}


\smallskip\noindent While relative completeness guarantees in Theorem~\ref{thm:completeness} are the first such guarantees for programs with non-determinism, they only apply to non-termina\-ting programs that contain an initial diverging configuration w.r.t.~some resolution of non-determinism. However, Example~\ref{ex:noninitialterm} shows that finding such a configuration might require using very high degree polynomials to resolve non-determinism, and in general such a configuration need not exist at all in non-termina\-ting programs. In order to ensure catching non-termination bugs in such examples, an algorithm with stronger guarantees is needed. To that end, we propose a modification of our algorithm for programs in which non-determinism appears only in branching. The new algorithm provides {\em stronger relative completeness guarantees} that can detect non-terminating behavior in programs with no initial diverging configurations or for which Check~1 is not practical, including the program in Example~\ref{ex:noninitialterm} (that is, its equivalent version in which non-determinism appears only in branching as we demonstrate in Example~\ref{ex:strongconst}).

To motivate this modification, let us look back at the conditions imposed on the predicate map $\BI$ by our algorithm. $\BI$ is required not to be an invariant, so that $\neg \BI$ contains a reachable configuration. However, this reachability condition cannot be encoded using polynomial constraints, so instead we require that $\neg\BI$ is not an inductive invariant, and then employ a safety prover which does not provide any guarantees. Our modification is based on the recent work of~\cite{AsadiCFGM20}, which presents a relatively complete method for reachability analysis in polynomial programs with non-determinism appearing only in branching.

\smallskip\noindent{\em Relatively complete reachability analysis.} We give a high level description of the method in~\cite{AsadiCFGM20}. Let $P$ be a program with non-determinism appearing only in branching, $\trsys$ its transition system, and $C$ a set of configurations defined by a propositional predicate map. The goal of the analysis is to check whether some configuration in $C$ is reachable in $\trsys$.

The witness for the reachability of $C$ in~\cite{AsadiCFGM20} consists of (1)~an initial configuration $\mathbf{c}$, (2)~a propositional predicate map $C^{\diamond}$ that contains $\mathbf{c}$, and (3)~a polynomial ranking function $f^C$ for $C^{\diamond}$ with respect to~$C$. A {\em polynomial ranking function} for $C^{\diamond}$ with respect to~$C$ is a map $f^C$ that to each location $\loc\in\locs$ assigns a polynomial expression $f^C(l)$ over program variables, such that each configuration $(l,\mathbf{x})\in C^{\diamond}\backslash C$ has a successor $(l',\mathbf{x'})\in C^{\diamond}$ with
\[f^C(l)(\mathbf{x})\geq f^C(l')(\mathbf{x}')+1\, \land\, f^C(l)(\mathbf{x})\geq 0, \]
where $C^{\diamond}$ and $C$ are treated as sets of configurations.
Intuitively, this means that for each configuration $(l,\mathbf{x})\in C^{\diamond}\backslash C$, the value of $f^C$ at this configuration is non-negative and there is a successor of this configuration in $C^{\diamond}$ at which the value of $f^C$ decreases by at least $1$. If the program admits such a witness, then we may exhibit a path from $\mathbf{c}$ to a configuration in $C$ by inductively picking either a successor in $C$ (and thus proving reachability), or a successor in $C^{\diamond}\backslash C$ along which $f^C$ decreases by $1$. As the value of $f^C$ in $\mathbf{c}$ is finite and $f^C$ is non-negative on $C^{\diamond}\backslash C$, decrease can happen only only finitely many times and eventually we will have to pick a configuration in $C$. It is further shown in~\cite{AsadiCFGM20} that any reachable $C$ admits a witness in the form of an initial configuration, a predicate map and a (not necessarily polynomial) ranking function.

For programs with non-determinism appearing only in branching, it is shown in~\cite{AsadiCFGM20} that all the defining properties of $\mathbf{c}$, $C^{\diamond}$ and $f^C$ can be encoded using polynomial constraints. Thus~\cite{AsadiCFGM20} searches for a reachability witness by introducing template variables for $\mathbf{c}$, $C^{\diamond}$ and $f^C$, encoding the defining properties using polynomial constraints and then reducing to constraint solving. The obtained constraints are at most quadratic in the template variables, as was the case in our algorithm for proving non-termination. Moreover, their analysis is relatively complete - if a witness of reachability in the form of an initial configuration $\mathbf{c}$, a propositional predicate map $C^{\diamond}$ and a polynomial ranking function $f^C$ exists, the method of~\cite{AsadiCFGM20} will find it.

\smallskip\noindent{\em Modification of our algorithm.} The modified algorithm is similar to Check 2, with only difference being that we encode reachability of $\neg\BI$ using polynomial constraints instead of requiring it not to be inductive in $\trsys$. The algorithm introduces a template of fresh variables determining $R^{\NA}$, $\tilde{I}$ and $\BI$. In addition, it introduces a template of fresh variables determining an initial configuration $\mathbf{c}$, a propositional predicate map $C^{\diamond}$ and a polynomial ranking function $f^{\neg\BI}$. The algorithm then imposes the following polynomial constraints:
\begin{compactitem}
	\item Encode the same conditions on $R^{\NA}$, $\tilde{I}$ and $\BI$ as in Check 2 to obtain $\Phi_{\textit{backward}}$.
	\item Introduce fresh variables $c_1,c_2,\dots,c_{|\vars|}$ defining the variable valuation of $\mathbf{c}$. Then substitute these variables into the assertion $\Theta_{init}$ specifying initial configurations in $\trsys$ to obtain the constraint $\phi_{\mathbf{c}}$ for $\mathbf{c}$ being an initial configuration.
	\item Fix a template for the propositional predicate map $C^{\diamond}$ of type-$(c,d)$ and maximal polynomial degree $D$. Encode that $C^{\diamond}$ contains $\mathbf{c}$ into the constraint $\phi_{\mathbf{c}, C^{\diamond}}$.
	\item For each location $\loc$ in $\trsys$, fix a template for a polynomial $f^{\neg\BI}(\loc)$ over program variables of degree at most $D$. That is, introduce a fresh template variable for each coefficient of such a polynomial.
	\item Using the method of~\cite{AsadiCFGM20}, for each locaiton $\loc$ encode the following condition
	\begin{equation*}
	\begin{split}
	&\forall \mathbf{x}. \mathbf{x}\models C^{\diamond}(\loc) \Rightarrow \mathbf{x}\in \neg\BI(\loc) \lor\, \Big( \Big(\bigvee_{\tau=(\loc,\loc',\rho_{\tau})}\mathbf{x}'\models C^{\diamond}(\loc')\land \\
	&\rho_{\tau}(\mathbf{x},\mathbf{x}')\land f^{\neg\BI}(\loc)(\mathbf{x})\geq f^{\neg\BI}(\loc')(\mathbf{x}')+1\Big) \land\, f^{\neg\BI}(\loc)(\mathbf{x})\geq 0\Big),
	\end{split}
	\end{equation*}
	as a polynomial constraint $\phi_{\loc,\textit{reach}}$. Note that, since we assume that non-determinism appears only in branching and not in variable assignments, for any $\mathbf{x}$ there is at most one variable valuation $\mathbf{x}'$ such that $\rho_{\tau}(\mathbf{x},\mathbf{x}')$ is satisfied. Thus, the above condition indeed encodes the condition that, if $(\loc,\mathbf{x})\not\in\neg\BI$, then at least one successor configuration satisfies the ranking function property. It is shown in~\cite{AsadiCFGM20} that this condition can be encoded into existentially quantified polynomial constraints over template variables, by using analogous semi-algebraic techniques that are used for inductive invariant generation in~\cite{ColonSS03, ChatterjeeFGG19} and which we use for invariant synthesis. We then take $\Phi_{\textit{reach}}= \wedge_{\loc}\,\phi_{\loc,\textit{reach}}$.
\end{compactitem}
The algorithm then tries to solve $\Phi_{\textit{modified}}=\Phi_{\textit{backward}}\land \phi_{\mathbf{c}} \land \phi_{\mathbf{c}, C^{\diamond}}\land \Phi_{\textit{reach}}$.

Soundness of the modified algorithm follows the same argument as the proof of Theorem~\ref{thm:soundness}. The following theorem presents the stronger relative completeness guarantees provided by the modified algorithm.

\begin{theorem}[Stronger relative completeness]\label{thm:strongerguarantees}
	Let $P$ be a program with polynomial integer arithmetic, in which non-determinism appears only in branching. Let $\trsys$ be its transition system. Suppose that $\trsys$ admits
	\begin{compactenum}
		\item a proper under-approximation $U$ restricting each non-deterministic assignment to a polynomial assignment,
		\item a propositional predicate map $\tilde{I}$ which is an inductive invariant in $\trsys$,
		\item a propositional predicate map $\BI$ which is an inductive backward invariant in $\trsys^{r,\tilde{I}(\locterm)}_{U}$, and
		\item a witness of reachability of $\neg\BI$ as in~\cite{AsadiCFGM20}.
	\end{compactenum}
	Then for high enough values of $c$, $d$ and $D$ bounding the template size for invariants and the polynomial degree, our algorithm proves non-termination of the program $P$.
\end{theorem}

\begin{rmk}
	The method of~\cite{AsadiCFGM20} encodes constraints for programs in which non-determinism appears only in branching, whereas in this work we talked about constraint encoding for programs in which non-determinism appears only in assignments. This is not an issue in the modified algorithm - we can always start with a program in which non-determinism appears only in branching to encode the reachability witness constraints, and then apply the trick from Section~\ref{sec:prelims} to replace each non-deterministic branching by an assignment.
\end{rmk}

\begin{example}\label{ex:strongconst}
	We show that the relative completeness guarantees of the modified algorithm apply to the program obtained from Fig.~\ref{fig:noninitialterm} by replacing the non-deterministic assignment of $u$ and the subsequent conditional branching with the non-deterministic branching given by \textbf{if $\ast$ then}. Specifically, the new program is obtained by removing the non-deterministic assignment of $u$ from the program, merging $\loc_1$ and $\loc_2$ in Fig.~\ref{ex:noninitialterm} into the new location $\loc_{1,2}$ and replacing the conditional by the non-deterministic branching. The reachability constraints for the modified algorithm are then encoded with respect to this new program. On the other hand, to encode the constraints as in Check~$2$, we consider the original program in Fig.~\ref{fig:noninitialterm}.
	
	To see that this program  satisfies the conditions of Theorem~\ref{thm:strongerguarantees}, we define $R^{\NA}$, $\tilde{I}$ and $\BI$ as in Example~\ref{ex:noninitialterm}.
	Then, one witness of reachability of $\neg\BI$ (where we identify $\loc_{1,2}$ with $\loc_1$) is defined by $\mathbf{c}=(\locinit,0,0,0)$,
	\begin{equation*}
	C^{\diamond}(\loc) = \begin{cases}
	(0\leq n\leq 99 \land b=0\land u=0) &\text{if $\loc\in\{\loc_0,\loc_{1,2}\}$}\\
	(0\leq n\leq 98 \land b=0\land u=0) &\text{if $\loc\in\{\loc_4,\loc_6\}$}\\
	(1\leq n\leq 99 \land b=0\land u=0) &\text{if $\loc=\loc_7$}\\
	(1\leq 0) &\text{otherwise};
	\end{cases}
	\end{equation*}
	and
	\begin{equation*}
	f^{\neg\BI}(\loc,n,b,u) = \begin{cases}
	5\cdot (100-n)+3 &\text{if $\loc=\loc_0$}\\
	5\cdot (100-n)+2 &\text{if $\loc=\loc_{1,2}$}\\
	5\cdot (100-n)+1 &\text{if $\loc=\loc_4$}\\
	5\cdot (100-n)+0 &\text{if $\loc=\loc_6$}\\
	5\cdot (100-n)+4 &\text{if $\loc=\loc_7$}\\
	0 &\text{otherwise};
	\end{cases}
	\end{equation*}
	To see that this is indeed the witness of reachability of $\neg\BI$, observe 
	$C^{\diamond}$ contains precisely the set of all configurations along the path from $\mathbf{c}=(\locinit,0,0,0)$ to the configuration $(\loc_1,99,0,0)$ in $\neg\BI$ that we described in Example~\ref{ex:noninitialterm} (recall, for reachability analysis we identify $\loc_1$ with $\loc_{1,2}$ in the modified program in which non-determinism appears only in branching), and that $f^{\neg\BI}$ is non-negative along this path and decreases by exactly $1$ in each step along the path.
\end{example}

\section{Experiments}\label{sec:experiments}

We present a prototype implementation of our algorithm in our tool $\RevTerm$.  Our implementation is available at the following link: \url{https://github.com/ekgma/RevTerm.git}.
We follow a standard approach to invariant generation \cite{ColonSS03,GulwaniSV08,ChatterjeeFGG19} which only fixes predicate map templates at cutpoint locations. For safety prover we use CPAchecker~\cite{BeyerK11} and for constraint solving we use three SMT-solvers: Barcelogic 1.2~\cite{BofillNORR08}, MathSAT5~\cite{mathsat5} and Z3~\cite{MouraB08}.

Since non-determinism in all our benchmarks appears in variable assignments only, we implemented only our main algorithm and not the modified algorithm with stronger guarantees for programs with branching-only non-determinism.

\smallskip\noindent{\em Benchmarks.} We evaluated our approach on benchmarks from the category 
{\em Termination of C-Integer Programs} of the Termination and Complexity Competition ($\TermComp$ \citep{GieslRSWY19}). The benchmark suite consists of 335 programs with non-determinism: 111 non-terminating, 223 terminating, and the Collatz conjecture for which termination is unknown. We compared $\RevTerm$ against the best state-of-the-art tools that participated in this category, namely $\AProVE$~\cite{GieslABEFFHOPSS17}, $\Ultimate$ \cite{ChenHLLTTZ18}, $\VeryMax$~\cite{BorrallerasBLOR17}, and also $\LoAT$~\cite{FrohnG19}.

\smallskip\noindent{\em Configurations of our tool.} Recall that our algorithm is parameterized by the template size for propositional predicate maps and the maximal polynomial degree. Also, it performs two checks which can be run sequentially or in parallel. Thus a {\em configuration} of $\RevTerm$ is defined by
	(a) the choice of whether we are running Check 1 or Check 2,
	(b) the template size $(c,d)$ for propositional predicate maps and the maximal polynomial degree $D$, and
	(c) the choice of an SMT-solver.
Our aim is to compare our algorithm to other existing approaches to non-termination proving and demonstrate generality of its relative completeness guarantees, ra\-ther than develop an optimized tool. Hence we test each configuration separately and count the total number of benchmarks that were proved to be non-terminating by at least one of the configurations. 
We consider configurations for both checks, each of the three SMT-solvers, and all template sizes in the set $\{(c,d,D)\mid 1\leq c\leq 5, 1\leq d \leq 5, 1\leq D\leq 2\}$.

\begin{center}
	\begin{table}[t]
		\centering
		\caption{Experimental results 
			with evaluation performed on the first platform. The NO/YES/MAYBE rows contain the total number of benchmarks which were proved non-terminating, terminating, or for which the tool proved neither, respectively. The next row contains the number of benchmarks proved to be non-terminating only by the respective tool. We also report the average and standard deviation (std.~dev.) of runtimes. The last two rows show the runtime statistics limited to successful non-termination proofs.}
		\begin{tabular}{|c||c|c|c|c|c|}
			\hline
			& $\RevTerm$ & $\Ultimate$ & $\VeryMax$ \\\hline\hline
			NO & 107 & 97 & 103 \\\hline
			YES & 0 & 209 & 213 \\\hline
			MAYBE & 228 & 29 & 19 \\\hline
			Unique NO & 3 & 1 & 0 \\\hline
			Avg.~time & 1.2s & 5.0s & 3.7s \\\hline
			Std.~dev. & 3.0s & 3.7s & 7.3s \\\hline
			Avg.~time for NO & 1.2s & 4.4s & 10.6s \\\hline
			Std.~dev. for NO & 3.0s & 3.8s & 9.4s \\\hline
		\end{tabular}
		\label{tab:exp1}
	\end{table}
\vspace{-1em}
\end{center}

\def \hfillx {\hspace*{-\textwidth} \hfill}

\begin{center}
	\begin{table*}[t]
		\begin{minipage}[t]{0.5\textwidth}
		\caption{Experimental results
		with evaluation performed on $\StarExec$~\cite{StumpST14}. The meaning of data is the same as in Table~\ref{tab:exp1}.}
		\begin{tabular}{|c||c|c|c|c|c|}
			\hline
			& $\RevTerm$ & $\LoAT$ & $\AProVE$ & $\Ultimate$ & $\VeryMax$ \\\hline\hline
			NO & 103 & 96 & 99 & 97 & 102 \\\hline
			YES & 0 & 0 & 216 & 209 & 212 \\\hline
			MAYBE & 232 & 239 & 20 & 29 & 21 \\\hline
			Unique NO & 2 & 1 & 0 & 0 & 0 \\\hline
			Avg.~time & 1.8s & 2.6s & 4.2s & 7.4s & 3.8s \\\hline
			Std.~dev. & 6.6s & 0.9s & 4.1s & 4.9s & 7.4s \\\hline
			Avg.~time NO & 1.8s & 2.6s & 5.0s & 7.0s & 10.8s \\\hline
			Std.~dev. NO & 6.6s & 0.9s & 3.9s & 7.1s & 9.5s \\\hline
		\end{tabular}
		\label{tab:exp2}
	\end{minipage}
	\hfillx
	\begin{minipage}[t]{0.43\textwidth}	
	\caption{Comparison of configurations based on which check they run and the SMT-solver used.
	}
		\begin{tabular}{|c||c|c|c|c|}
			\hline
			& Barcelogic & MathSAT5 & Z3 & Total \\
			& 1.2 & & & \\\hline\hline
			Check 1 & 84 & 98 & 80 & 103 \\\hline
			Check 2 & 69 & 54 & 63 & 74 \\\hline
			Total & 96 & 98 & 82 & 107 \\\hline
			
		\end{tabular}
		\label{tab:exp3}
	\end{minipage}
	\end{table*}
\vspace{-1em}
\end{center}

\noindent{\em Experimental results.} Our experiments were run on two platforms, and we include the results for each of them in separate tables. The first platform is Debian, 128 GB RAM, Intel(R) Xeon(R) CPU E5-1650 v3 @ 3.50GHz, 12 Threads. The experimental results are presented in Table~\ref{tab:exp1} and the timeout for each experiment was 60s.

 We could not install the dependencies for $\AProVE$ on the first platform and $\LoAT$ does not support the input format of benchmarks, so we also evaluate all tools except for $\LoAT$ on $\StarExec$~\cite{StumpST14} which is a platform on which $\TermComp$ was run. We take the results of the evaluation of $\LoAT$ on $\StarExec$ from~\cite{FrohnG19} which coupled it with $\AProVE$ for conversion of benchmarks to the right input format. Note however that the solver Barcelogic 1.2 is not compatible with $\StarExec$ so the number of non-terminations $\RevTerm$ proves is smaller compared to Table~\ref{tab:exp1}. The experimental results are presented in Table~\ref{tab:exp2}, and the timeout for each experiment was 60s. The timeout in both cases is on wallclock time and was chosen to match that in~\cite{FrohnG19}. We note that in $\TermComp$ the timeout was 300s and $\Ultimate$ proved 100 non-terminations, whereas $\AProVE$ and $\VeryMax$ proved the same number of non-terminations as in Table~\ref{tab:exp2}.

 From Tables~\ref{tab:exp1} and~\ref{tab:exp2} we can see that $\RevTerm$ outperforms other tools in terms of the number of proved non-terminations. The average time for $\RevTerm$ is computed by taking the fastest successful configuration on each benchmark, so the times indicate that running multiple configurations in parallel would outperform the state-of-the-art. Since $\AProVE$, $\Ultimate$ and $\VeryMax$ attempt to prove either termination or non-termination of programs, we include both their average times for all solved benchmarks and for non-termination proofs only.

\smallskip\noindent {\em Performance by configuration.} We now discuss the performance of each configuration based on whether it runs Check 1 or Check 2 and based on which SMT-solver it uses. For the purpose of this comparison we only consider evaluation on the first platform which supports Barcelogic 1.2. Comparison of configurations in terms of the total number of solved benchmarks is presented in Table~\ref{tab:exp3}. We make two observations:
\begin{compactitem}
	\item Configurations using Check 1 prove 103 out of 112 non-terminations, which matches the performance of all other tools. This means that the relative completeness guarantees provided by our approach are quite general.
	\item Even though some SMT-solvers perform well and solve many benchmarks, none of them reaches the number 107. This means that our performance is dependent on the solver choice and designing a successful tool would possibly require multiple solvers. For example, from our results we observed that MathSAT5 performs particularly well for Check 1 with templates of small size ($c,d\in\{1,2\}$), while Barcelogic 1.2 is best suited for templates of larger size (with $c\geq 3$) and for Check 2. While this could be seen as a limitation of our approach, it also implies that our algorithm would become even more effective with the improvement of SMT-solvers.
\end{compactitem}
 Finally, in Appendix~\ref{app:exp} we present a comparison of configurations based on the template sizes for propositional predicate maps. A key observation there is that for any benchmark that $\RevTerm$ proved to be non-terminating, it was sufficient to use a template for predicate maps with $c\leq 3$, $d\leq 2$ and $D\leq 2$. This implies that with a smart choice of configurations, it suffices to run a relatively small number of configurations which if run in parallel would result in a tool highly competitive with the state-of-the-art.

\section{Related Work}\label{sec:relatedwork}


\medskip\noindent{\em Non-termination proving.} A large number of techniques for proving non-termination consider {\em lasso-shaped} programs, \\
which consist of a finite prefix (or stem) followed by a single loop  without branching~\cite{GuptaHMRX08,LeikeH18}. Such techniques are suitable for being combined with termination provers~\cite{HarrisLNR10}. Many modern termination provers repeatedly generate traces which are then used to refine the termination argument in the form of a ranking function, either by employing safety provers~\cite{CookPR06} or by checking emptiness of automata~\cite{HeizmannHP14}. When refinement is not possible, a trace is treated like a lasso program and the prover would try to prove non-termination. However, lassos are not sufficient to detect {\em aperiodic} non-termina\-tion, whereas our approach handles it. Moreover, programs with nested loops typically contain infinitely many lassos which may lead to divergence, and such methods do not provide relative completeness guarantees.

TNT~\cite{GuptaHMRX08} proves non-termination by exhaustively searching for candidate lassos. For each lasso, it searches for a {\em recurrence set} (see Section~\ref{sec:certificate}) and this search is done via constraint solving. The method does not support non-determi\-nism. 

Closed recurrence sets (see Section~\ref{sec:certificate}) are a stronger notion than the recurrence sets, suited for proving non-termina\-tion of non-deterministic programs. The method for computing closed recurrent sets in~\cite{ChenCFNO14} was implemented in $\mathsf{T2}$ and it uses a safety prover to eliminate terminating paths iteratively until it finds a program under-approximation and a closed recurrence set in it. The method can detect aperiodic non-termination. However it is likely to diverge in the presence of many loops, as noted in~\cite{LarrazNORR14}.


The method in \cite{LarrazNORR14} was implemented in $\VeryMax$~\cite{BorrallerasBLOR17} and it searches for witnesses to non-termination in the form of quasi-invariants, which are sets of configurations that cannot be left once they are entered. Their method searches for a quasi-invariant in each strongly-connected subgraph of the program by using Max-SMT solving. Whenever a quasi-invariant is found, safety prover is used to check its reachability. The method relies on multiple calls to a safety prover and does not provide relative completeness guarantees.

$\AProVE$~\cite{GieslABEFFHOPSS17} proves non-termination of Java programs~\cite{BrockschmidtSOG11} with non-determinism. It uses constraint solving to find a recurrence set in a given loop, upon which it checks reachability of the loop. The key limitation of this approach is that for programs with nested loops for which the loop condition is not a loop invariant, it can only detect recurrence sets with a single variable valuation at the loop head.

An orthogonal approach to recurrence sets was presented in~\cite{LeikeH18}. It considers lasso-shaped programs with linear arithmetic and represents infinite runs as geometric series. Their method provides relative completeness guarantees for the case of deterministic lasso-shaped programs. It also supports non-determinism, but does not provide relative completeness guarantees. The method has been implemented as a non-termination prover for lasso traces in $\Ultimate$~\cite{ChenHLLTTZ18}.

The method in~\cite{UrbanGK16} tries to prove either termination or non-termination of programs with non-determinism by making multiple calls to a safety prover. For each loop, a termination argument is incrementally refined by using a safety prover to sample a terminating trace that violates the argument. Once such terminating traces cannot be found, a safety prover is again used to check the existence of non-terminating traces in the loop.

The work of \cite{GulwaniSV08} considers deterministic programs with linear integer arithmetic. They present a constraint solving-based method for finding the {\em weakest liberal precondition (w.l.p.)} of a fixed propositional predicate map template. They then propose a method for proving non-termination which computes the w.l.p.~for the postcondition "false", and then checks if it contains some initial configuration. While this approach is somewhat similar to Check 1, encoding and solving the {\em weakest} precondition constraints of a given template is computationally expensive and unnecessary for the purpose of proving non-termination. In Check 1, we do not impose such a strict condition. Moreover, initial diverging configurations are not sufficient to prove non-termination of non-deterministic programs. It is not immediately clear how one could use w.l.p.~calculus to find a diverging configuration within a loop, like in Example~\ref{ex:noninitialterm}.

The tool $\mathsf{Invel}$~\cite{VelroyenR08} proves non-termination of Java programs using constraint solving and heuristics to search for recurrence sets. It only supports deterministic programs. In \cite{LeQC15} a Hoare-style approach is developed to infer sufficient preconditions for terminating and non-terminating behavior of programs. As the paper itself mentions, the approach is not suitable for programs with non-determinism.

While all of the methods discussed above are restricted to programs with linear arithmetic, the following two methods also consider non-linear programs.

The tool $\mathsf{Anant}$~\cite{CookFNO14} proves non-termination of programs with non-linear arithmetic and heap-based operations. They define {\em live abstractions}, which over-approximate a program's transition relation while keeping it sound for proving non-termination. Their method then over-approximates non-linear assignments and heap-based commands with non-determini\-stic linear assignments using heuristics to obtain a live abstraction with only linear arithmetic. An approach similar to~\cite{GuptaHMRX08} but supporting non-determinism is then used, to exhaustively search for lasso traces and check if they are non-termina\-ting. The over-approximation heuristic they present is compatible with our approach and could be used to extend our method to support operations on the heap.

$\LoAT$~\cite{FrohnG19} proves non-termination of integer programs by using loop acceleration. If a loop cannot be proved to be non-terminating, the method tries to accelerate it in order to find paths to other potentially non-terminating loops.

\section{Conclusion and Future Work}\label{sec:conclusion}

We present a new approach for proving 
non-termination of polynomial programs with a relative completeness guarantee.
For programs that do not satisfy this guarantee, our approach requires safety provers.
An interesting direction of future work would be to consider approaches that can present stronger completeness guarantees. Another interesting direction would be to consider usefulness of the program reversal technique to studying other properties in programs.


\section*{Acknowledgements}
This research was partially supported by the ERC CoG 863818 (ForM-SMArt) and the Czech Science Foundation grant No. GJ19-15134Y.

\bibliography{cav20}

\clearpage
\appendix

\begin{center}
	{\Large Appendix}
\end{center}

\section{Omitted Proofs}\label{app:certificate}

\subsection{Sound and Complete Certificate for Non-termination}

\begin{lemma}
	Let $P$ be a non-terminating program and $\trsys$ its transition system. Then there exist a proper under-approximation $U$ of $\trsys$ and a closed recurrence set $\mathcal{C}$ in $U$.
\end{lemma}

\begin{proof}
	Since $P$ is non-terminating, $\trsys=(L,\vars,\locinit,\Theta_{init},\transitions)$ admits an under-approximation $U'=(L,\vars,\locinit,\Theta_{init},\transitions^{U'})$ and a closed recurrence set $\mathcal{C}$ in $U'$. Define the under-approximation $U=(L,\vars,\locinit,\Theta_{init},\transitions^U)$ of $\trsys$ by defining $\rho^U_\tau$ for each $\tau=(l,l',\rho_{\tau})\in\transitions$ as follows:
	\begin{equation*}
	\rho^U_{\tau} = \rho^{U'}_{\tau} \cup \{(\mathbf{x},\mathbf{x}')\mid (\mathbf{x},\mathbf{x}')\in \rho_{\tau}, (l,\mathbf{x})\not\in \mathcal{C}\}.
	\end{equation*}
	Thus we extend $\rho^{U'}_{\tau}$ in order to make each configuration outside $\mathcal{C}$ have the same set of successors as in $\transitions$.
	
	\smallskip\noindent We claim that the under-approximation $U$ of $\trsys$ is proper and that $\mathcal{C}$ is a closed recurrence set in $U$, proving the lemma. Indeed, by the construction of $U$ every configuration outside $\mathcal{C}$ which has a successor in $\trsys$ also has at least one successor in $U$. For a configuration in $\mathcal{C}$ this is immediate since its set of successors in $U$ is the same in $U'$ and $\mathcal{C}$ is a closed recurrence set in $U'$. Hence $U$ is proper. To see that $\mathcal{C}$ is a closed recurrence set in $U$, note again that every configuration in $\mathcal{C}$ has the same successor set in $U'$ and in $U$, thus it has at least one successor in $U$ which is in $\mathcal{C}$, and also all of its successors in $U$ are in $\mathcal{C}$. Therefore, as $\mathcal{C}$ contains an initial configuration in $\trsys$ and thus in $U$, the claim follows.
\end{proof}

\begin{theorem}[Soundness of our certificate]
	Let $P$ be a program and $\trsys$ its transition system. If there exists a $\BI$-certificate $(U,\BI,\Theta)$ in $\trsys$, then $P$ is non-terminating.
\end{theorem}

\begin{proof}
	Consider a predicate map $\neg\BI$. Note that in a transition system defined by a program, every configuration has at least one successor. So as $U$ is proper, every configuration in $\neg\BI$ has at least one successor in $U$. On the other hand, by Theorem~\ref{lemma:inductiverreverse} we have that $\neg\BI$ is inductive in $U$ since $\BI$ is inductive in $U^{r,\Theta}$. Hence for every configuration in $\neg\BI$, all of its successors in $U$ are also in $\neg\BI$.
	
	\smallskip\noindent We also know that $\BI$ is not an invariant in $\mathcal{\trsys}$, so there exists a reachable configuration $\mathbf{c}$ in $\trsys$ which is contained in $\neg\BI$. We use it to construct a non-terminating execution. Pick any finite path $\mathbf{c}_0,\mathbf{c}_1,\dots,\mathbf{c}_k=\mathbf{c}$ in $\trsys$, which exists by reachability. Since terminal configurations do not have non-terminal successors, none of the configurations along this finite path is terminal. Then, starting from $\mathbf{c}$ we may inductively pick successors in $U$ which are in $\neg\BI$. This is possible from what we showed above. Moreover, none of the picked successors can be terminal all reachable terminal configurations are contained in $\BI$. Thus we obtain a non-terminating run, as wanted.
\end{proof}

\begin{theorem}[Complete characterization of non-termination]
	Let $P$ be a non-terminating program with transition system $\trsys$. Then $\trsys$ admits a proper under-approximation $U$ and a predicate map $\BI$ such that $\BI$ is an inductive backward invariant in the reversed transition system $U^{r,\mathbb{Z}^{|\vars|}}$, but not an invariant in $\trsys$.
\end{theorem}

\begin{proof}
	Since $P$ is non-terminating, from Lemma~\ref{lemma:properun} we know that $\trsys$ admits a proper under-approximation $U$ and a closed recurrence set $\mathcal{C}$ in $U$. For each location $l$ in $\trsys$, let $\mathcal{C}(l)=\{\mathbf{x}\mid (l,\mathbf{x})\in \mathcal{C}\}$. Define the predicate map $\BI$ as $\BI(l)=\neg \mathcal{C}(l)$ for each $l$. We claim that $\BI$ is an inductive backward invariant in $U^{r,\mathbb{Z}^{|\vars|}}$, but not an invariant in $\trsys$.
	
	\smallskip\noindent By definition of closed recurrence sets, $\BI$ contains all terminal configurations in $\trsys$, i.e.~initial configurations of $U^{r,\mathbb{Z}^{|\vars|}}$. On the other hand, for every configuration in $\mathcal{C}$ all of its successors in $U$ are also contained in $\mathcal{C}$. Thus, as a predicate map $\mathcal{C}$ is inductive, hence by Theorem~\ref{lemma:inductiverreverse} $\BI=\neg \mathcal{C}$ is inductive in $U^{r,\mathbb{Z}^{|\vars|}}$. This shows that $\BI$ is an inductive backward invariant in $U^{r,\mathbb{Z}^{|\vars|}}$. On the other hand, as a closed recurrence set $\mathcal{C}$ contains an initial configuration in $\trsys$, $\BI$ does not contain all reachable configurations and is thus not an invariant in $\trsys$. The claim of the theorem follows.
\end{proof}

\subsection{Soundness of the Algorithm for Non-termination proving}

\begin{theorem}[Soundness]
	If Algorithm~\ref{algo:nonterm} outputs ''Non-ter\-mination'' for some input program $P$, then $P$ is non-terminating.
\end{theorem}

\begin{proof}
	If the algorithm outputs ''Non-termination'' for an input program $P$, then it was either able to show that $\Phi_1$ is feasible for $P$, or that $\Phi_2$ is feasible for $P$ with the subsequent safety check being successful.
	
	Suppose that the algorithm showed that $\Phi_1$ is feasible and found a resolution of non-determinism $R^{\NA}$, an initial configuration $\mathbf{c}$ in $\trsys$, and a type-$(c,d)$ propositional predicate map $I$ satisfying properties in Check 1 of the algorithm. We claim that $(\trsys_{R^{\NA}},\neg I,\mathbb{Z}^{|\vars|})$ is a $\BI$-certificate for $P$ and thus, by Theorem~\ref{thm:certsoundness}, $P$ is non-terminating. By Definition~\ref{def:restriction}, $R^{\NA}$ defines a proper under-approximation $\trsys_{R^{\NA}}$ of $\trsys$. On the other hand, $I$ is inductive for $\trsys_{R^{\NA}}$, so by Theorem~\ref{lemma:inductiverreverse}, $\neg I$ is inductive for $\trsys_{R^{\NA}}^{r,\mathbb{Z}^{|\vars|}}$. Moreover, $I(\locterm)=\emptyset$, so $\neg I(\locterm)=\mathbb{Z}^{|\vars|}$ contains all terminal configurations and $\neg I$ is an inductive backward invariant for $\trsys_{R^{NA}}^{r,\mathbb{Z}^{|\vars|}}$. Finally, $\neg I$ does not contain the initial configuration $\mathbf{c}$ in $\trsys$ and is thus not an invariant for $\trsys$. 
	
	Suppose now that $\Phi_2$ was shown to be feasible and that the subsequent safety check was successful. Then the algorithm had to find a resolution of non-determinism $R^{\NA}$, a type-$(c,1)$ predicate map $\tilde{I}$, a type-$(c,d)$ predicate map $\BI$, and a transition $\tau$ in $\trsys$ satisfying constraints in Check 2 of the algorithm. We claim that $(\trsys_{R^{\NA}},\BI,\tilde{I}(\locterm))$ is a $\BI$-certificate for $P$ and thus, by Theorem~\ref{thm:certsoundness}, $P$ is non-terminating. Again, by Definition~\ref{def:restriction}, $R^{\NA}$ defines a proper under-approxi\-mation $\trsys_{R^{\NA}}$ of $\trsys$. $\BI$ is an inductive backward invariant for $\trsys_{R^{\NA}}^{r,\tilde{I}(\locterm)}$ and $\tilde{I}(\locterm)$ contains all terminal configurations in $\trsys$, as $\tilde{I}$ is an invariant for $\trsys$. Finally, since the safety check was successful there exists a reachable configuration in $\trsys$ contained in $\neg\BI$, which shows that $\BI$ is not an invariant for $\trsys$.
\end{proof}

\lstset{language=affprob}
\lstset{tabsize=2}
\newsavebox{\exapppppp}
\begin{lrbox}{\exapppppp}
	\centering
	\begin{lstlisting}[mathescape]
	$l_0$:	while $x\geq 1$ do
	$l_1$:		$y := 10 \cdot x$
	$l_2$:		while $x\leq y$ do
	$l_3$:			$x := x+1$
	    od
	   od
	\end{lstlisting}
\end{lrbox}

\begin{figure}[t]
	\centering
	\usebox{\exapppppp}
	\caption{Example illustrating aperiodic non-termination. }
	\label{fig:aperiodic}
\end{figure}

\section{Missing Table for the Experiments}\label{app:exp}

\begin{center}
	\begin{table}[hb]
		\centering
		\caption{Comparison of configurations based on the template size for predicate maps. A cell in the table corresponding to $(C=i,D=j)$ contains the number of benchmarks that were proved to be non-terminating by a configuration using template size $(c,d)$ with $c \leq i$ and $d\leq j$.}
		\begin{tabular}{|c||c|c|c|c|c|}
			\hline
			& $D=1$ & $D=2$ & $D=3$ & $D=4$ & $D=5$ \\\hline\hline
			$C=1$ & 58 & 76 & 77 & 78 & 78 \\\hline
			$C=2$ & 90 & 97 & 97 & 97 & 97 \\\hline
			$C=3$ & 102 & 107 & 107 & 107 & 107 \\\hline
			$C=4$ & 103 & 107 & 107 & 107 & 107 \\\hline
			$C=5$ & 103 & 107 & 107 & 107 & 107 \\\hline
		\end{tabular}
		\label{tab:exp4}
	\end{table}
\end{center}

\section{Example for Aperiodic Non-termination}\label{app:aperiodic}

\begin{example}\label{ex:aperiodic}
Consider the program in Fig.~\ref{fig:aperiodic}. Note that non-determinism in this program is implicit and appears in assigning the initial variable valuation. Moreover, every initial configuration defines a unique execution which is terminating if and only if the initial value of $x$ is non-positive. \\
We claim that every non-terminating execution in this program is aperiodic. Consider an execution starting in an initial configuration $(\loc_0,x_1,y_1)$ with $x_1\geq 1$. First, note that the inner program loop is terminating thus the execution will execute the outer loop infinitely many times. The $i$-th iteration of the outer loop in this execution is of the form $\loc_0,\loc_1,(\loc_2,\loc_3)^{(9x_i)}$ where $(\loc_0,x_i,y_i)$ the configuration in this execution upon starting the $i$-th iteration of the outer loop. A simple inductive argument shows that $x_i=10^{i-1}$ for each $i\in\mathbb{N}$. This proves the claim that each non-terminating execution is indeed aperiodic. \\
We now show that Check 1 of our algorithm can prove that this program is non-terminating. Since there are no non-deterministic assignments, the resolution of non-determinism $R^{\NA}$ is trivial. Consider the initial configuration $\mathbf{c}=(\loc_0,1,1)$ and a propositional predicate map $I$ defined via $I(\loc) = (x\geq 1)$ if $\loc\in\{\loc_0,\loc_1,\loc_2,\loc_3\}$, and $I(\locterm)=\emptyset$. Then $R^{\NA}$, $I$ and $\mathbf{c}$ satisfy all the properties in Check 1 of our algorithm, and our algorithm can prove non-termination of the program in Fig.~\ref{fig:aperiodic}.
\end{example}

\pagebreak

\section{Reversed Transition System for the Program in Example~\ref{ex:noninitialterm}}\label{app:reversed}

\begin{figure}[hb]
	\begin{center}
	\centering
	\begin{tikzpicture}
	\node[ran] (while) at (0,0)  {$\loc_0$};
	\node[ran, right = 3.6cm of while] (term) {$ \locterm $};
	\node[ran, below = 1.5 cm of while] (assignu) {$ \loc_1 $};
	\node[ran, below = 1.5 cm of assignu] (assignb) {$ \loc_2 $};
	\node[ran, below left = 2.4 cm of assignb] (assignbmin) {$ \loc_3 $};
	\node[ran, below = 3 cm of assignb] (assignbzer) {$ \loc_4 $};
	\node[ran, below right = 2.4 cm of assignb] (assignbplu) {$ \loc_5 $};
	\node[ran, below = 3 cm of assignbzer] (assignn) {$ \loc_6 $};
	\node[ran, below = 1.5 cm of assignn] (if) {$ \loc_7 $};
	\node[ran, below = 1.5 cm of if] (while2) {$ \loc_8 $};
	\node[ran, left = 2.6 cm of while2] (skip) {$ \loc_9 $};
	
	\draw[tran] (term) to node[font=\scriptsize,draw, fill=white, 
	rectangle,pos=0.5] {$(b!=0 \lor n\geq 100)\land I_{n,u,b}$} (while);
	\draw[tran] (if) to node[font=\scriptsize,draw, fill=white, 
	rectangle,pos=0.5] {$n'=n-1\land I_{u,b}$} (assignn);
	\draw[tran] (assignb) to node[font=\scriptsize,draw, fill=white, 
	rectangle,pos=0.5] {$u=1\land I_{n,b}$} (assignu);
	\draw[tran] (assignbmin) to node[font=\scriptsize,draw, fill=white, 
	rectangle,pos=0.5] {$u\leq -1\land I_{n,u,b}$} (assignb);
	\draw[tran] (assignbzer) to node[font=\scriptsize,draw, fill=white, 
	rectangle,pos=0.5] {$u=0\land I_{n,u,b}$} (assignb);
	\draw[tran] (assignbplu) to node[font=\scriptsize,draw, fill=white, 
	rectangle,pos=0.5] {$u\geq 1\land I_{n,u,b}$} (assignb);
	\draw[tran] (assignn) to node[font=\scriptsize,draw, fill=white, 
	rectangle,pos=0.5] {$b=1\land I_{n,u}$} (assignbplu);
	\draw[tran] (assignn) to node[font=\scriptsize,draw, fill=white, 
	rectangle,pos=0.5] {$b=0\land I_{n,u}$} (assignbzer);
	\draw[tran] (assignn) to node[font=\scriptsize,draw, fill=white, 
	rectangle,pos=0.5] {$b=-1\land I_{n,u}$} (assignbmin);
	\draw[tran] (assignu) to node[font=\scriptsize,draw, fill=white, 
	rectangle,pos=0.5] {$(b=0\land x\leq 99)\land I_{n,u,b}$} (while);
	\draw[tran] (if) to node[font=\scriptsize,draw, fill=white, 
	rectangle,pos=0.5] {$(b\geq 1\land x=100)\land I_{n,u,b}$} (while2);
	\draw[tran] (while2) to node[font=\scriptsize,draw, fill=white, 
	rectangle,pos=0.5] {$I_{n,u,b}$} (skip);
	
	\node[left = 4cm of while, circle, minimum size = 3mm] (dum) {};
	\draw[tran, rounded corners] (while) -- (dum.east) -- node[font=\scriptsize,draw, fill=white, 
	rectangle,pos=0.5] {$(x!=100 \lor b\leq 0)\land I_{n,u,b}$} (dum.east|-if) -- (if);
	
	\node[below = 0.8cm of skip, circle, minimum size = 3mm] (dum2) {};
	\node[below = 0.8cm of while2, circle, minimum size = 3mm] (dum3) {};
	\draw[tran, rounded corners] (skip) -- (dum2.north) -- node[font=\scriptsize,draw, fill=white, 
	rectangle,pos=0.5] {$I_{n,u,b}$} (dum3.north) -- (while2);
	\end{tikzpicture}
	\caption{Reversed transition system of the program in Fig.~\ref{fig:noninitialterm} with the resolution of non-determinism that assigns the constant expression $1$ to the non-deterministic assignment of the variable $u$. For readability, we use $I_{n,u,b}$ to denote $n'=n\land u'=u\land b'=b$, $I_{u,b}$ to denote $u'=u\land b'=b$, etc. }
	\vspace{-1.5em}
	\label{fig:revnoninitialterm}
	\end{center}
\end{figure}
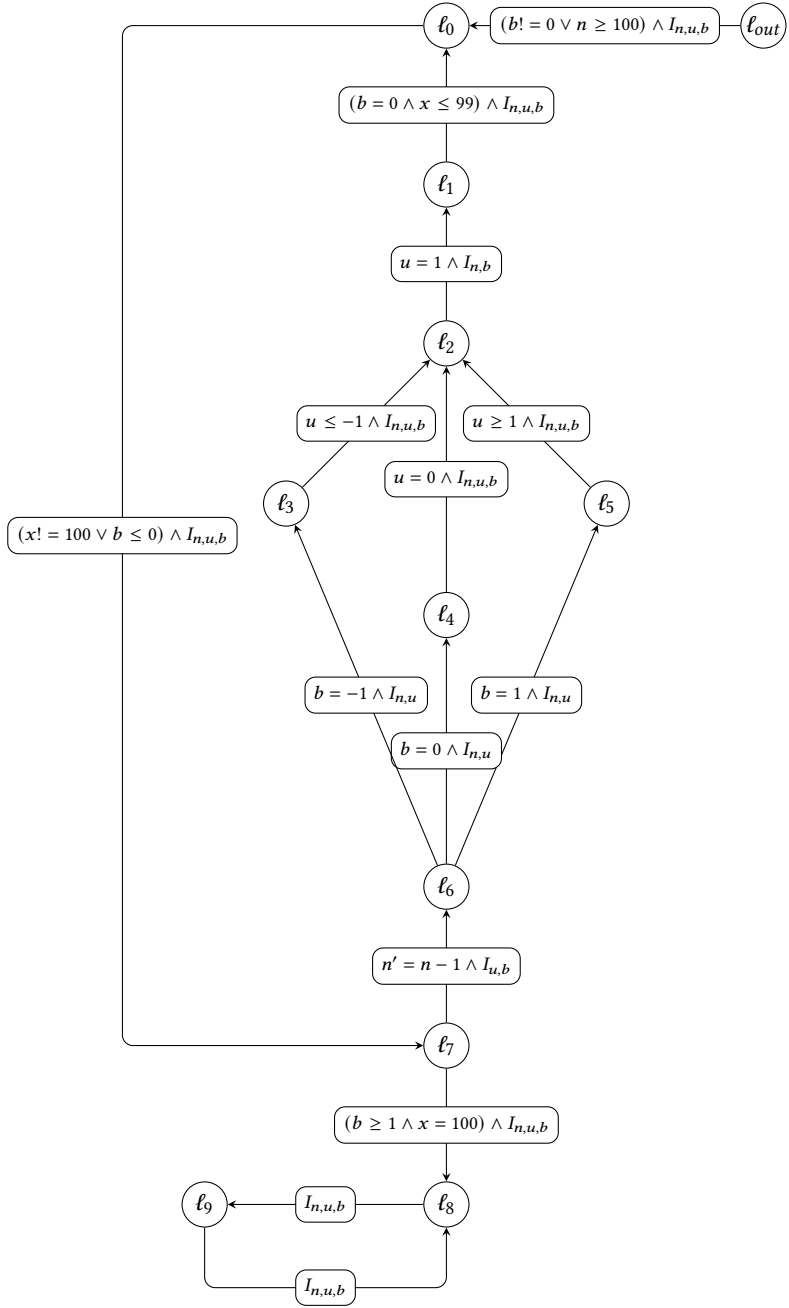

\end{document}